%% file: main.tex
\pdfoutput = 1
\documentclass[11pt, dvipsnames]{article}
\usepackage[margin=1in]{geometry}
\usepackage{setspace}
\usepackage[round,sort&compress]{natbib}
\let\cite\citep
\usepackage{amsmath}
\usepackage{amssymb}
\usepackage{algorithm}
\usepackage{varwidth}
\usepackage{booktabs}
\usepackage{algpseudocode}
\allowdisplaybreaks
\usepackage{hyperref}
\makeatletter
\newcommand{\printfnsymbol}[1]{%
  \textsuperscript{\@fnsymbol{#1}}%
}
\makeatother
\usepackage{subcaption, makecell}
\usepackage{graphicx}
\usepackage{multirow}
\usepackage{graphicx}
\usepackage{wrapfig}
\usepackage{caption}
\usepackage{natbib}
\usepackage{subcaption}
\usepackage{color}
\usepackage{bbm}
\usepackage{float}

\usepackage{amsthm}
\usepackage{dsfont}
\usepackage{mathrsfs}
\usepackage{lmodern}
\usepackage{hyperref}

\usepackage{mathtools}
\usepackage{verbatim}
\usepackage{subcaption}
\usepackage{enumitem}
\usepackage{amsthm}
\usepackage{macros}

\usepackage{booktabs}
\newtheorem{theorem}{Theorem}[section]
\newtheorem{lemma}[theorem]{Lemma}
\newtheorem{proposition}[theorem]{Proposition}

\newtheorem{remark}[theorem]{Remark}

\newtheorem{definition}[theorem]{Definition}
\usepackage{xcolor}
\usepackage{makecell} 

\usepackage{booktabs} 
\usepackage{array} 

\providecommand{\keywords}[1]
{
  \small	
  \textbf{\textit{Key Words and Phrases:}} #1
}

\usepackage{appendix}
\title{Privacy Preserving Mechanisms for Coordinating Airspace Usage in Advanced Air Mobility}

\author{Chinmay Maheshwari\thanks{Chinmay Maheshwari, Maria G. Mendoza, and Victoria Tuck contributed equally to this work and are listed alphabetically.}, Maria G. Mendoza\footnotemark[1], Victoria Marie Tuck\footnotemark[1], Pan-Yang Su,\\ Victor L Qin, Sanjit A. Seshia, Hamsa Balakrishnan, Shankar Sastry} 

\date{}

\begin{document}

\maketitle

\renewcommand\thefootnote{}
\noindent\footnotetext{Authors’ Contact Information: Chinmay Maheshwari,
Maria G. Mendoza, Victoria Tuck, Pan-Yang Su: \{chinmay\_maheshwari, maria\_mendoza, victoria\_tuck, pan\_yang\_su\}@berkeley.edu, University of California, Berkeley; Victor L Qin, victorqi@mit.edu, MIT School of Engineering; Sanjit A. Seshia,
sseshia@eecs.berkeley.edu, University of California, Berkeley; Hamsa Balakrishnan, hamsa@mit.edu, MIT School of Engineering; Shankar Sastry, sastry@coe.berkeley.edu, University of California, Berkeley.}
\renewcommand\thefootnote{\arabic{footnote}}

\thispagestyle{empty}
\pagestyle{empty}

\begin{abstract}
\input{Sections/abstract}
\end{abstract}

\newpage
\input{Sections/introduction}
\input{Sections/Model}
\section{Algorithmic Design of Auction Mechanism without Private Valuations}\label{sec: Algorithm}
\input{Sections/Algorithms}

\section{Drone Delivery in Toulouse: A Case Study}\label{sec: DroneDelivery}
\input{Sections/Numerics}

\input{Sections/SupplementaryResultsMain}\label{sec: SupplementalDrone}

\section{Limitations}\label{sec: limitations}
\input{Sections/Limitations}


\section{Conclusions}\label{sec: Conclusion}
\input{Sections/Conclusion}

\section{Acknowledgments}
Chinmay Maheshwari, Pan-Yang Su, and Shankar Sastry acknowledge support from NSF Collaborative Research: Transferable, Hierarchical, Expressive, Optimal, Robust, Interpretable NETworks (THEORINET) Award No. DMS 2031899. Victoria Tuck and Shankar Sastry acknowledge support from Provably Correct Design of Adaptive Hybrid Neuro-Symbolic Cyber-Physical Systems, Defense Advanced Research Projects Agency award number FA8750-23-C-0080. Maria G. Mendoza and Shankar Sastry acknowledge support from NASA, Clean Sheet Airspace Operating Design project award number MFRA2018-S-0471. Victor L Qin was supported by the National Science Foundation Graduate Research Fellowship Program under Grant No. 2141064. Any opinions, findings, conclusions, or recommendations expressed in this material are those of the author(s) and do not necessarily reflect the views of the National Science Foundation.

\bibliography{refs}
\bibliographystyle{plainnat}
\newpage
\appendix
\section{A Simple Example}\label{sec:section21}
\input{Sections/TimeExtendedGraph}

\section{Proof of Theoretical Results}\label{sec: ProofsAll}
\subsection{Proof of Proposition \ref{prop: MarketClearingExists}} \label{sec:Proposition33}
    \input{Appendix/ProofExistence}

\subsection{Proof of Proposition \ref{prop: FixedPointOptimization}}\label{sec:Proposition35}
    \input{Appendix/ProofComputation}

\section{Derivation of Inner Loop Updates in Algorithm \ref{algo:1}} \label{sec:ADMM}
\input{Appendix/ADMM_generalMethod}

\section{Vertiport Reservation Mechanism in Northern California}\label{sec:SanFranciscoCase}

\input{Appendix/Boards}

\section{Table of Notations} \label{sec:appendixA}
\input{Appendix/TableOfNotations}

\addtolength{\textheight}{-12cm}   


\end{document}

%% file: Sections/abstract.tex
Advanced Air Mobility (AAM) operations are expected to transform air transportation while challenging current air traffic management practices. By introducing a novel market-based mechanism, we address the problem of on-demand allocation of capacity-constrained airspace to AAM vehicles with heterogeneous and private valuations. We model airspace and air infrastructure as a collection of contiguous regions (or sectors) with constraints on the number of vehicles that simultaneously enter, stay, or exit each region. Vehicles request access to airspace with trajectories spanning multiple regions at different times. We use the graph structure of our airspace model to formulate the allocation problem as a path allocation problem on a time-extended graph. To ensure that the cost information of AAM vehicles remains private, we introduce a novel mechanism that allocates each vehicle a budget of ``air-credits'' (an artificial currency) and anonymously charges prices for traversing the edges of the time-extended graph. We seek to compute a competitive equilibrium that ensures that: (i) capacity constraints are satisfied, (ii) a strictly positive resource price implies that the sector capacity is fully utilized, and (iii) the allocation is integral and optimal for each AAM vehicle given current prices, without requiring access to individual vehicle utilities. However, a competitive equilibrium with integral allocations may not always exist. We provide sufficient conditions for the existence and computation of a fractional-competitive equilibrium, where allocations can be fractional. Building on these theoretical insights, we propose a distributed, iterative, two-step algorithm that: 1) computes a fractional competitive equilibrium, and 2) derives an integral allocation from this equilibrium. We validate the effectiveness of our approach in allocating trajectories for the emerging urban air mobility service of drone delivery.

%% file: Sections/introduction.tex
\section{Introduction}
The emergence of advanced air mobility (AAM) operations, including urban air mobility (UAM) and unmanned aerial vehicles (UAVs), is expected to transform the landscape of the air transportation system. These new aerial platforms can provide air taxi services that better connect rural and suburban communities with urban centers, facilitate package and medical deliveries, and support infrastructure and public safety \cite{10472699}.

Integrating advanced air mobility (AAM) into existing air traffic management (ATM) systems presents complex and unresolved challenges. Projections of AAM traffic density and operational complexity have raised concerns about the scalability of traditional ATM infrastructure. The Federal Aviation Administration (FAA) has acknowledged these challenges \cite{faaUtmConops, faaUamConops}, stating:
\begin{quote} ``{Given the number, type, and duration of Unmanned Aircraft System (UAS) operations envisioned, the existing Air Traffic Management (ATM) System infrastructure and associated resources cannot cost-effectively scale to deliver services for UAS.}'' \hfill --- \textit{FAA UAS/UTM Con-Ops} (\cite{faaUtmConops})
\end{quote}

The limitations of conventional ATM approaches become evident when examining their design principles. Existing systems \cite{doi:10.1287/opre.46.3.406, doi:10.1287/trsc.34.3.239.12300, doi:10.1287/opre.1100.0899, 4282854, 10.1007/978-3-642-86726-2_17,BALL2018186, doi:10.1287/trsc.2020.0995} are built to manage fixed-wing aircraft operating between established airports, with flight schedules planned weeks or months in advance to minimize overall delays. In contrast, AAM introduces a fundamentally different paradigm, with a high volume of electric Vertical Take-Off and Landing (eVTOL) aircraft and UAVs operating on demand and pursuing diverse objectives. These vehicles will not only travel between fixed vertiports but also serve ad-hoc destinations, such as residential areas for package deliveries, further straining legacy ATM systems.
To address this challenge, the FAA \cite{faaUtmConops, faaUamConops} and other Air Navigation Service Providers (ANSPs) worldwide \cite{SesarUamConops} have stated that daily traffic management for AAM operations will be delegated to third-party service providers (SPs). These providers will coordinate directly with AAM vehicles to allocate airspace efficiently and safely, reducing reliance on the FAA.

Beyond operational challenges, the introduction of third-party SPs in AAM systems raises additional concerns, particularly regarding privacy. A key issue arises from the heterogeneous private valuations of AAM vehicles, which can vary significantly depending on their specific use cases. For instance, a passenger air taxi may have strict scheduling constraints, whereas a regional cargo flight might be more flexible with delays \cite{skorup2019auctioning, seuken2022market}. This variability in preferences, coupled with the sensitivity of business or personal data, makes AAM operators reluctant to disclose their private valuation information to SPs.

Against this backdrop, this work aims to answer the following question:
\begin{quote}
How can SPs allocate capacity-constrained airspace resources to dynamically arriving AAM vehicles with heterogeneous private valuations in a way that allows AAM vehicles to achieve an (approximately) optimal allocation based on their valuations, without requiring them to disclose this private information to the SPs?
\end{quote}
We address this question by focusing on a single service provider managing access to a specific region of airspace, as even under a single SP, efficiently allocating airspace to AAM vehicles without requiring their private information remains an open problem.

We model the airspace as a set of contiguous regions, each with specific capacity constraints on the number of AAM vehicles (modeled as eVTOLs) that can arrive, depart, or remain in a given region at any time (see Fig. \ref{fig:AirSegementationFigure}). Vehicles submit requests for airspace access through a menu of feasible (discrete-time) time-trajectories (or air corridors), where each time-trajectory corresponds to a sequence of tuples specifying the time step and the sector they wish to occupy at that time.
We formulate the allocation of these time-trajectories within capacity-constrained airspace as a path allocation problem on a time-extended graph (Definition \ref{def: Time-extendedGraph}), where all capacity constraints are represented as constraints on the graph's edges.

\begin{figure}
    \centering
    \includegraphics[width=0.6\linewidth]{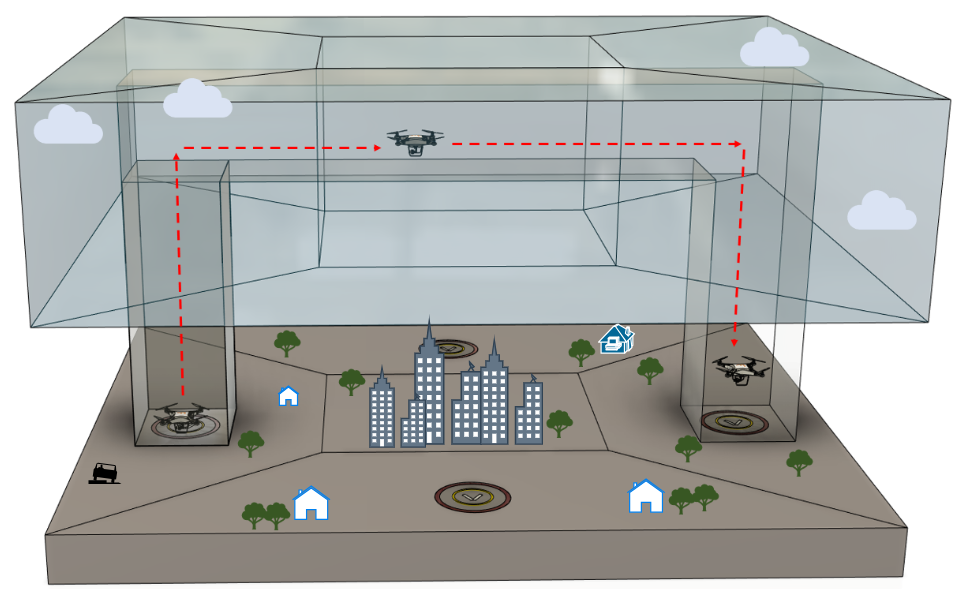}
    \caption{Model of airspace as a set of contiguous regions, each having arrival, departure, and transit constraints. Some regions are at the cruising altitude, while others encompass vertiport/launch pads.}
    \label{fig:AirSegementationFigure}
\end{figure}

To prevent the monopolization of airspace by major players, we introduce an artificial currency-based auction mechanism \cite{skorup2019auctioning}. Furthermore, to accommodate the dynamically arriving requests of AAM vehicles, we propose implementing this auction mechanism in a receding-horizon manner (see Section \ref{sec: MarketsAAM}). This approach periodically collects AAM vehicle requests and determines allocations using the proposed auction mechanism.

In each of these auction mechanisms, the SP allocates “air-credits” (the artificial currency) to each AAM vehicle requesting airspace access and charges an anonymous price (in air-credits) for using different airspace regions (i.e., resources on the time-extended graph). Based on these prices, each vehicle selects its most preferred time-trajectory on the time-extended graph, maximizing its valuation while adhering to its budget constraint.
The SP's objective is to design prices and allocate airspace efficiently and safely, guided by the following desiderata:
\((i)\) Given the price vector, the SP’s allocation should be optimal for every AAM vehicle, ensuring that each vehicle maximizes its private valuation subject to budget constraints; \((ii)\) The capacity constraints of the airspace must be respected;  \((iii)\) Prices must be nonnegative, and if strictly positive, the capacity of each airspace region must match its demand. An allocation-price tuple satisfying these conditions is known as a \textit{competitive equilibrium} in economics, which may not always exist \cite{7c65302b-f079-361a-94f1-0c3c9f6fc76b}. However, drawing inspiration from Fisher markets under linear constraints \cite{jalota2023fisher}, we establish the existence of a \textit{fractional competitive equilibrium}—a relaxation of the competitive equilibrium that permits fractional allocations (Proposition \ref{prop: MarketClearingExists}). Moreover, we demonstrate that the prices at a fractional competitive equilibrium can be computed as the optimal dual multipliers of a \textit{budget-adjusted welfare problem} (see \eqref{eq: PlannerOptMain}), which is a convex optimization problem (Lemma \ref{lem: ConvexOptMain}). Notably, the budget adjustment for each vehicle is determined by the optimal dual multiplier associated with a linear constraint of this optimization problem. Consequently, computing a fractional competitive equilibrium reduces to solving a fixed-point problem (Proposition \ref{prop: FixedPointOptimization}).

Building on these theoretical insights, we propose a \textit{two-step algorithmic procedure} for allocating AAM vehicles to airspace. In the first step, we develop a \textit{two-loop} algorithm to compute the fractional competitive equilibrium without requiring information about the vehicles' private valuations. Specifically, this step involves solving the fixed-point problem stated in Proposition \ref{prop: FixedPointOptimization} using a two-loop algorithm (see Algorithm \ref{algo:1}) that mimics fixed-point iteration.
The inner loop solves a reformulated budget-adjusted welfare problem (cf. \eqref{eq: PlannerOptADMM}) in a distributed manner using the Alternating Direction Method of Multipliers (ADMM) (see Appendix \ref{sec:ADMM} for a review of ADMM). This ensures that AAM vehicles do not need to share their valuations with SP and other AAM vehicles. The outer loop then updates the budget adjustment parameter using the latest value of the dual multiplier associated with each AAM vehicle's individual constraints.

Algorithm \ref{algo:1} can be interpreted as an online learning process, where AAM vehicles iteratively refine their trajectory choices based on anonymized market signals—such as expected demand and prices—while the SP dynamically adjusts these signals based on observed demand by AAM vehicles. This enables the SP to estimate equilibrium prices without access to the private valuations of AAM vehicles, while vehicles deconflict through indirect price-mediated coordination, obviating the need for direct trajectory or valuation sharing. This mechanism aligns with online learning of market mechanisms (cf. \cite{kleinberg2003value, li2020simple, bistritz2021online}), as both the agents and the SP update their strategies sequentially based on observed feedback.

In the second step (i.e., Algorithm \ref{algo:2}), we derive an integral allocation from the fractional competitive equilibrium obtained in the first step while keeping the prices unchanged. We rank the vehicles according to the fractional allocation they received for their most desired resource in the first step. The SP then allocates resources to the vehicles sequentially according to this ranking, updating the remaining capacity after each allocation (Algorithm \ref{algo:2}).
Importantly, in both Algorithms \ref{algo:1}-\ref{algo:2}, the SP requires only information on resource demands, feasible time trajectories, and the most desired path of each AAM vehicle, without accessing any private valuation data. Likewise, individual AAM vehicles do not obtain information about the time trajectories or private valuations of other vehicles, ensuring that privacy is preserved throughout the process.

To validate the effectiveness of our approach, we analyze drone-based package delivery using a dataset of drone trajectories generated with realistic physical models by Airbus over the city of Toulouse, France. A further study on the scheduling problem for electric air taxis on a hypothesized air traffic network in Northern California, United States can be found in Appendix \ref{sec:SanFranciscoCase}.

\subsection{Related Works}
\input{Sections/RelatedWorks}

\subsection{Organization}

The article is organized as follows: In Section \ref{sec: ModelAAM} we introduce the model of airspace management studied in this paper. In Section \ref{sec: MarketsAAM} we present a high-level overview of our approach, which implements an artificial-currency-based auction mechanism in a receding horizon manner. Next, in Section \ref{sec: Auction_Single_Full_Information} we formally present such an auction mechanism along with theoretical results on fractional-competitive equilibrium. In Section \ref{sec: Algorithm}, we present the algorithmic procedure to compute an approximate market mechanism. In Section \ref{sec: DroneDelivery}, we validate the performance of our mechanism using a drone delivery dataset generated based on a real drone dynamics model from Airbus.
We discuss the limitations of our approach in Section \ref{sec: limitations}, and conclude this study in Section \ref{sec: Conclusion} with some interesting directions for future work. Proofs and additional supplementary material are provided in the Appendix. In Appendix \ref{sec:section21}, we present a simple example to illustrate the time-extended graph and constraints \eqref{eq: SelectingRoute}-\eqref{eq: PathConstraints}. In Appendix \ref{sec: ProofsAll}, we provide proofs for all the theoretical results discussed in this paper. A detailed explanation of the ADMM formulation can be found in Appendix \ref{sec:ADMM}. In Appendix \ref{sec:SanFranciscoCase}, we study an additional AAM scenario of vertiport reservation for air-taxi services in Northern California. Finally, in Appendix \ref{sec:appendixA}, we consolidate important notations used in this work in a table.

%% file: Sections/RelatedWorks.tex
The literature on market mechanisms for airspace management in Advanced Air Mobility (AAM) remains relatively limited \cite{seuken2022market}. Some recent works \cite{balakrishnan2022cost, chin2023traffic} have explored airspace management through second-price auctions combined with congestion management algorithms, as well as combinatorial auctions \cite{doi:10.2514/6.2024-4454}. However, these approaches are restricted to unit-capacity regions, limiting their applicability to real-world AAM systems that require flexible allocation across multiple capacity-constrained airspace sectors.
Su et al. \cite{su2024incentivecompatiblevertiportreservationadvanced} addressed these limitations by utilizing a generalized Vickrey–Clarke–Groves (VCG) auction for AAM resource management, incorporating considerations of social welfare, safety, and congestion. Their approach ensures proportional fairness by optimizing a social cost function based on the weighted utilities of all fleet operators. A key design goal in their paper is to elicit truthful preferences from AAM vehicles to achieve efficient airspace allocation. However, this reliance on truthful bidding raises privacy concerns, as vehicles must disclose their exact valuations through bids, potentially exposing sensitive operational information to competitors or regulators.
Balakrishnan and Chandran \cite{balakrishnan2017distributed} proposed a column generation algorithm that iteratively updates prices to determine an allocation that satisfies capacity constraints. However, their method relies on vehicles reporting their private valuations, which may raise privacy concerns. In contrast, our approach eliminates the need for AAM vehicles to disclose their private valuations, enhancing privacy while still achieving efficient allocation.

A distinguishing feature of our approach compared to existing market mechanisms for AAM is the use of artificial currency. While monetary transactions are effective in eliciting preferences, as noted by \cite{seuken2022market}, they can disproportionately advantage operators with greater financial resources. Our approach mitigates this issue by introducing a system of artificial currency that helps ensure fairness.
The idea of using artificial currency for fair and efficient resource allocation has been extensively studied in economics, beginning with \cite{VARIAN197463}. These works typically consider environments where agents are endowed with artificial currency that holds no value outside the market. However, most of these studies focus on the allocation of substitutable goods, whereas our setting requires agents to select multiple goods over a time-extended network, leading to complementarities rather than substitutability. Artificial currency mechanisms that handle complementarities have been studied in \cite{7c65302b-f079-361a-94f1-0c3c9f6fc76b, 10.1145/3580507.3597809}. The key distinction of our work is that we allow agents to save currency for future use, and the saved budget influences the agent's utility in a quasi-linear manner. This feature aligns our model more closely with combinatorial auction environments.

Combinatorial auctions enable participants to bid on bundles of items rather than individual items \cite{cramton2006introduction}. Given the exponential complexity inherent in these auctions, there is no single format that universally applies across all settings. In environments with budget constraints, iterative auction formats—such as the Simultaneous Ascending Auction, the Ascending Proxy Auction, and the Clock-Proxy Auction \cite{10.1162/154247604323068168, cramton2006combinatorial}—are particularly relevant. 
These formats allow agents to observe current prices and iteratively adjust their bids within their budget constraints before final submission. 
Additionally, they offer a privacy advantage by enabling incremental bid submission, thereby reducing exposure of full preference information. 
However, these mechanisms may suffer in efficiency in combinatorial auction settings with strong complementarities. In our approach, we propose a new way to model complementarities using linear equality constraints and leverage recent advances in Fisher markets with linear constraints to determine allocation. 
In particular, we compute a fractional-competitive equilibrium, which provides a relaxed solution that is easier to compute. From this fractional-competitive equilibrium, we then derive an integral allocation that satisfies the capacity constraints in the problem.

%% file: Sections/Model.tex
\section{Model of Advanced Air Mobility}\label{sec: ModelAAM}
Consider the problem of allocating airspace to Advanced Air Mobility (AAM) vehicles, such as drones and eVTOLs, which enable novel AAM services in urban environments. 
In our model, we segment the urban airspace into contiguous regions or sectors, denoted by \(\mathcal{R}\).
The spatial configuration of the airspace is represented as a graph \( \mathcal{G} = (\mathcal{R}, \mathcal{E}) \), where \( \mathcal{R} \) corresponds to the set of vertices, and \( \mathcal{E} \subseteq \mathcal{R} \times \mathcal{R} \) represents the set of edges that connect adjacent regions, indicating feasible movements for AAM vehicles between regions.
To account for the temporal dimension, we divide the entire day into \(T\) time-steps (with each time-step comprising of \(\tau\) seconds of time). AAM vehicles arrive dynamically, each requesting access to the airspace.

Each AAM vehicle has a feasible set of time-trajectories (also referred to as ``air corridors'') that it can utilize within the airspace. 
Each vehicle can independently determine the set of feasible trajectories by accounting for energy consumption, travel time, and other operational factors. 
Furthermore, the feasible set only includes time-trajectories that start before the vehicle's takeoff; mid-flight trajectory changes are not permitted.
Fig. \ref{fig: AuxGraph} illustrates a schematic representation of a time-trajectory for a package delivery scenario. 
For every feasible trajectory, the AAM vehicle generates a \emph{private valuation} that may differ across trajectories, reflecting its preferences over trajectories.

In our model, we assume that each region has a limit on the number of vehicles that can simultaneously arrive, depart, or remain in that region at any given time\footnote{This constraint arises from safety concerns that limit the number of vehicles that can be autonomously de-conflicted in a confined space \cite{BAURANOV2021100726}.}. Due to capacity constraints, it may not be feasible to allocate each AAM vehicle its \emph{most preferred} time-trajectory, as this could violate airspace constraints.
In such a scenario, a service provider (SP) is typically responsible for managing the airspace and assigning each AAM vehicle a feasible time-trajectory while ensuring compliance with airspace constraints.
To facilitate this process, we introduce the framework of a \emph{time-extended graph}, which is essential for integrating arrival, departure, and transit constraints when allocating time-trajectories to AAM vehicles.
\begin{figure*}[!ht]
        \centering
    \begin{minipage}[l]{0.9\textwidth}
        \centering
\includegraphics[width=\textwidth]{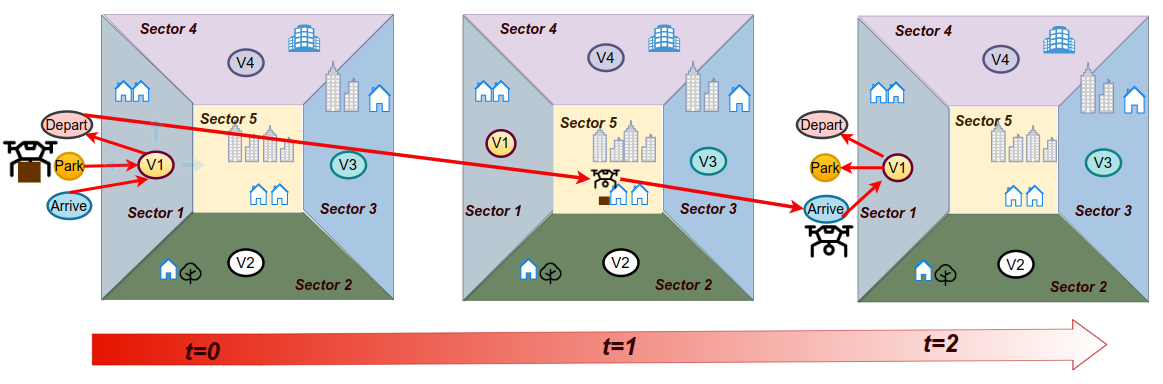}
    \end{minipage}
      \caption{A time-trajectory of a drone delivering a package in an urban setting. The drone starts from the launch pad \(\textsf{V1}\) (Sector 1) and needs to drop a package in Sector 5 before returning. Here, we have shown a simple trajectory that moves between regions in one time-step, but in general, such trajectories can remain in any region for multiple time-steps. 
      } 
      \label{fig: AuxGraph}
\end{figure*}

\begin{definition}[Time-extended Graph]\label{def: Time-extendedGraph}
We define \(\tilde{\mc{G}} = (\tilde{\mc{R}}, \tilde{\mc{E}})\) as the \emph{time-extended graph} with horizon \(T\), for some positive integer \(T\), such that 
\begin{itemize}
        \item[(i)] \(\tilde{\mc{R}} = \cup_{t=1}^{T}\cup_{r\in\mc{R}} \{\nu(r,t), \nu^\arr(r,t), \nu^\dep(r,t)\}\), where \(\nu(r,t)\), \(\nu^\arr(r,t)\) , and \(\nu^\dep(r,t)\) are three replicas of region \(r\) at time \(t\). 
    \item[(ii)] 
\(\tilde{\mc{E}}=\cup_{j=1}^{4}\tilde{\mc{E}}^{(j)}\subseteq \tilde{\mc{R}}\times\tilde{\mc{R}},\) where

\begin{itemize}
    \item \(\tilde{\mc{E}}^{(1)} := \cup_{t=1}^{T}\{(\nu^\arr(r,t),\nu(r,t))\}\). Any edge of the type \((\nu^\arr(r,t),\nu(r,t))\) has capacity\footnote{The arrival and departure constraints are primarily required for airspace regions comprising of vertiports/launch-pads.} \(C^\arr(r,t)\).  
    \item \(\tilde{\mc{E}}^{(2)} := \cup_{t=1}^{T}\{(\nu(r,t),\nu^\dep(r,t))\}\). Any edge of  the type \((\nu(r,t),\nu^\dep(r,t))\) has capacity \(C^\dep(r,t)\). 
    \item \(\tilde{\mc{E}}^{(3)} := \cup_{t=1}^{T-1}\{(\nu(r,t),\nu(r,t+1))\}\). Any edge of the type \((\nu(r,t),\nu(r,t+1))\) has capacity \(C^\park(r,t)\).
     \item \(\tilde{\mc{E}}^{(4)} := \cup_{t=1}^{T-1}\cup_{(r,r')\in \mc{E}}\{(\nu^\dep(r,t),\nu^\arr(r',t+1))\}\). Any edge of the type \((\nu^\dep(r,t),\nu^\arr(r',t+1))\) is unconstrained.  
\end{itemize} 
    \end{itemize}
\end{definition}
Any time-trajectory of an AAM vehicle is a path on this time-extended graph. A simple example describing the time-extended graph and time-trajectories is provided in Appendix \ref{sec:section21}.

In this work, we propose a market-based mechanism for the service provider (SP) to allocate capacity-constrained airspace infrastructure to dynamically arriving AAM vehicles. Our mechanism ensures that: \((i)\) capacity constraints of the airspace are strictly satisfied, \((ii)\) each AAM vehicle receives an (approximately) optimal allocation according to its preferences, and  \((iii)\) AAM vehicles are not required to disclose their private valuations to the SP. The detailed design and implementation of this mechanism are discussed in the following section.

\section{High-level Overview of Market-based Mechanism}\label{sec: MarketsAAM}

We propose an auction-based approach that allocates airspace to dynamically arriving AAM vehicles in a \emph{receding-horizon} fashion by \textit{periodically} allocating the available airspace capacity through an auction mechanism. In particular, 
over a total duration of \(T\) time-steps (with each time-step comprising of \(\tau\) seconds of time), the service provider (SP) conducts \(I\) auctions. In every auction, the SP allocates the airspace to AAM vehicles participating in that auction. 

\begin{wrapfigure}{r}{0.65\textwidth} 
    \centering
\includegraphics[width=0.6\textwidth]{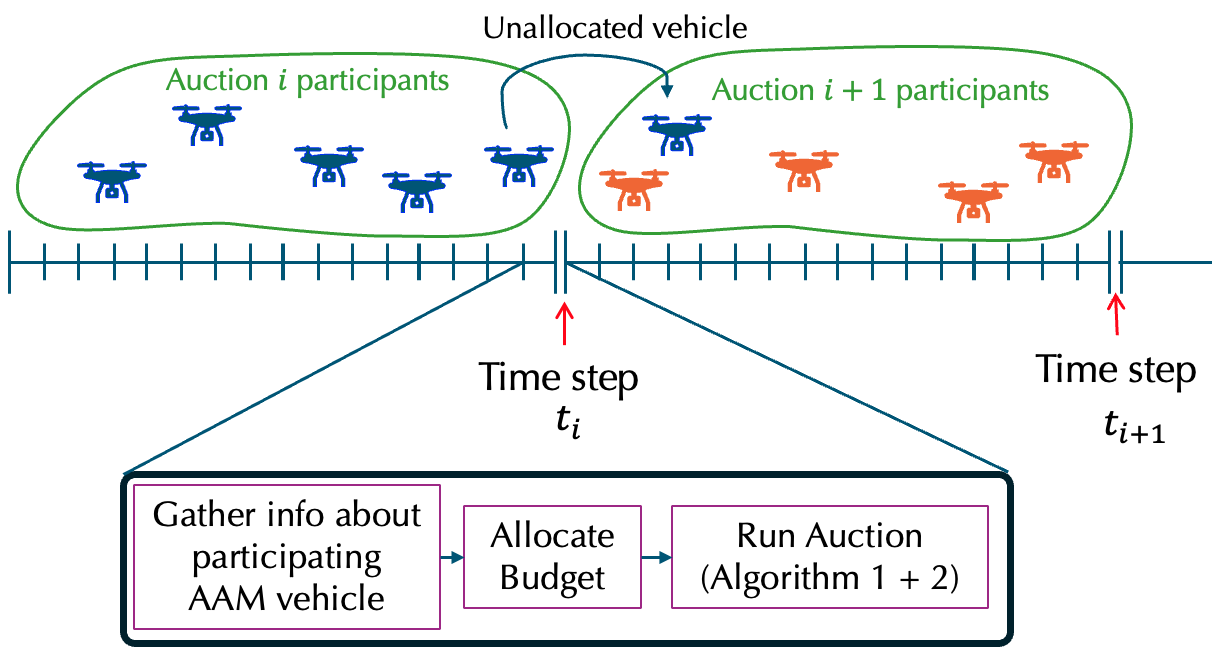}
\caption{A schematic depiction of the receding horizon approach.}\label{fig:recedingHorizonSchematic}
\end{wrapfigure}

The allocation of the \(i\)-th auction is determined at time-step 
\(
t_i = (i-1)\lfloor T/I \rfloor + 1.
\) The set of vehicles participating in the $i$-th auction comprises of new vehicles that have requested airspace access to SP after the \((i-1)\)-th auction (i.e. after time-step $t_{i-1}$) but also those that were not allocated in previous auctions. 
Each AAM vehicle has a feasible set of time-trajectories that it can follow\footnote{Note that the time-trajectory allocated in auction \(i\) may start at time-step \(t_i\) and can end at time-step \(T\).}. In each auction, the SP assigns a time-trajectory to each participating AAM vehicle from its feasible set\footnote{ Note that the feasible set of trajectories of each vehicle include an option \(\varnothing\) which indicates that the vehicle will not take-off in that auction (i.e remain unallocated), which makes sure that our problem always has a feasible allocation.}. It is important to 
emphasize that the flight trajectory of each AAM vehicle is finalized before takeoff, and they are not allowed to participate in subsequent auctions to modify their trajectory mid-flight. Finally, the SP 
updates the remaining airspace capacity before initiating the next auction. A schematic of the receding-horizon approach is provided in Fig. \ref{fig:recedingHorizonSchematic}.

At a high level, in each auction, the SP allocates a certain amount of ``air credits'' to each participating AAM vehicle. These air credits, along with other credits they have from past auctions, act as an artificial currency for purchasing airspace access during that auction.
Additionally, the SP imposes a payment (in air credits) for the use of each edge in the time-extended graph. AAM vehicle can utilize any leftover budget in future auctions to purchase access to the airspace.

The main design component of this auction mechanism is to determine these prices in a way that allows each AAM vehicle to afford an (approximately) optimal time-trajectory—one that maximizes its utility within its allocated budget—while ensuring that the overall airspace allocation adheres to capacity constraints.
Moreover, these prices must be computed in a manner that preserves the privacy of AAM vehicles, ensuring that they do not need to disclose their private valuations. 
In Section \ref{sec: Auction_Single_Full_Information}, we formally present one such auction mechanism that is used in the receding horizon approach discussed above, assuming that the SP has access to private valuations. Additionally, we study the theoretical properties of the proposed auction mechanism. 
In Section~\ref{sec: Algorithm}, we relax the assumption that the SP knows the private valuations of AAM vehicles and develop an algorithmic approach (Algorithm \ref{algo:1}-\ref{algo:2}) to implement the proposed auction mechanism without requiring this knowledge.

\section{Auction Mechanism: Design and Analysis}\label{sec: Auction_Single_Full_Information}
In this section, we formally describe the elements of our proposed auction mechanism, along with theoretical guarantees, assuming that the service provider (SP) has access to the private valuations of each AAM vehicle.  

Let \(U\) be the set of AAM vehicles requesting access to the airspace in the current auction. Each AAM vehicle \( u \in U \) is allocated a budget of air credits, denoted by \( w_u \geq 0 \)\footnote{The variation in budgets among AAM vehicles can arise due to two factors: (i) savings accumulated from previous auctions, and (ii) the priority given by SP, for instance in case of disaster relief and emergency service vehicles.}. 
The set of feasible time-trajectories for any AAM vehicle \( u \in U \) is given by \( M_u = R_u \cup \{\varnothing\} \), where \( R_u \) represents a subset of paths on the time-extended graph (cf. Definition \ref{def: Time-extendedGraph}), and \( \varnothing \) denotes the option to drop out of the system if no feasible path is available due to high congestion.

We define \( x_{u,e} \in \{0,1\} \) to indicate whether an AAM vehicle \( u \) is using edge \( e \in \tilde{\mathcal{E}} \), and \( x_{u,\varnothing} \in \{0,1\} \) to indicate whether the AAM vehicle \( u \) has dropped out of the system. 
Furthermore, each AAM vehicle has the option to convert its unused budget into an ``outside option'' for future auctions\footnote{ The budget converted into the outside option is carried forward and added to the AAM vehicle's budget in future auctions.}. 
We use \( x_{u,\outsideOpt} \) to denote the amount of outside options that the AAM vehicle consumes in any auction. For concise notation, we define \( \mathbf{x}_u = (x_{u,e})_{e\in\tilde{\mathcal{E}}} \) and \(\bar{\mathbf{x}}_u = \begin{bmatrix} \mathbf{x}_u^\top, x_{u,\outsideOpt}, x_{u,\varnothing} \end{bmatrix}^\top.\) 

The utility derived by any vehicle \(u\in U\) from selecting a route \(s\in R_u\) is denoted by \(v_{u,s}\in \mathbb{R}_+\), selecting the option \(\varnothing\) is denoted by \(v_{u,\varnothing}\in \mathbb{R}_+\), and per unit consumption of outside option is given by \( v_{u,\outsideOpt} \in \mathbb{R}_{+}\). Therefore, the overall utility derived by AAM vehicle \(u\) is given by 
\begin{align}\label{eq: UtilityUser}
f_u(\bar{\mathbf{x}}_u) 
= \sum_{s\in R_u}v_{u,s} x_{u, e^\ast(s)} + v_{u,\outsideOpt}x_{u,\outsideOpt} + v_{u,\varnothing}x_{u, \varnothing}, 
\end{align}
where \(e^\ast(s)\) denotes the \textit{departing edge}\footnote {In our framework, we use the convention that the AAM vehicles place all their valuation of a route on the first edge on the time-extended graph that goes out from the origin node. This is an edge of type \(\tilde{\mathcal{E}}^{(4)}\) in Definition \ref{def: Time-extendedGraph}. Additionally, we add a constraint (cf. \eqref{eq: Cap_ConstraintM}) that ensures that if a vehicle selects a departing edge corresponding to a route, then all edges on that route will be selected}. on the route \(s\).

The SP charges a price \(p_e,\) for any AAM vehicle using the edge \(e\in \tilde{\mathcal{E}}\), and a payment \( p_\outsideOpt \geq 0 \) for the consumption per unit of the outside option.

Upon observing the prices, each AAM vehicle solves the following optimization problem: 
\begin{subequations}\label{eq: UserOpt}
\begin{align}    \max_{\bar{\mathbf{x}}_u} &\quad f_u(\bar{\mathbf{x}}_u)  \tag{IOP}\label{eq: IOP_Objective}\\ 
    \text{s.t.} & \quad \mathbf{p}^\top \mathbf{x}_u + p_{\outsideOpt}x_{u,\outsideOpt} \leq w_u 
    \label{eq: BudgetConstraints}\\ 
    {} & \quad \tilde{\mathbf{a}}_u^\top \mathbf{x}_u  + x_{u,\varnothing} = 1 \label{eq: SelectingRoute}\\
{} & \quad \tilde{\mathbf{A}}_{u}\mathbf{x}_u = \mathbf{0}
   \label{eq: PathConstraints} \\
    {} & \quad \textbf{x}_u\in \{0,1\}^{|\tilde{\mathcal{E}}|}, x_{u,\varnothing} \in \{0,1\},\label{eq: IntegralConstraints}
\end{align}
\end{subequations}
where \(\tilde{\mathbf{a}}_u \in \mathbb{R}^{|\tilde{\mathcal{E}}|}\) is such that the constraint \eqref{eq: SelectingRoute} enforces \(\sum_{s\in R_u}x_{u,e^\ast(s)} + x_{u,\dropOut} = 1\), indicating that the AAM vehicle \(u\) will either select a path in \(R_u\) or will drop out. The matrix  \(\tilde{\mathbf{A}}_{u,s}\in \mathbb{R}^{K \times (|\tilde{\mc{E}}|)}\) in \eqref{eq: PathConstraints} represents two types of constraints: (i) \(\textbf{A}_{u}\mathbf{x}_u = \mathbf{0}\), where \(\textbf{A}_{u}\in \mathbb{R}^{|\tilde{\mathcal{R}}|\times|\tilde{\mathcal{E}}|}\) is an incidence matrix of the time-extended graph encoding flow-balance constraints at each node in \(\tilde{\mathcal{G}}\); and (ii)  \(\mathbf{B}_{u,s}\mathbf{x}_u = 0\) for each \(s\in R_u,\) where \(\mathbf{B}_{u,s}\in\mathbb{R}^{(K-|\tilde{\mathcal{R}}|)\times |\tilde{\mathcal{E}}|}\) encodes the constraint that the flow on any edge connecting two different regions along path \(s\) matches the flow on the departing edge \(e^\ast(s)\). Intuitively, (ii) ensures that any feasible solution that satisfies \eqref{eq: SelectingRoute} and (i) results in a unique edge flow. We present a simple example in Appendix \ref{sec:section21} to describe these constraints.

In \eqref{eq: UserOpt}, \eqref{eq: IOP_Objective} represents the utility derived by the AAM vehicle \(u\); \eqref{eq: BudgetConstraints} denotes the budget constraint of AAM vehicle \(u\); \eqref{eq: SelectingRoute} represents the requirement that AAM vehicle \(u\) must select at least one of the paths in \(R_u\) or consider dropping out; \eqref{eq: PathConstraints} ensures a unique edge flow for every feasible solution from \eqref{eq: SelectingRoute}; and \eqref{eq: IntegralConstraints} ensures that the selections made by AAM vehicles are integral. 
Note that the feasible set in \eqref{eq: UserOpt} is non-empty. This is because \(\mathbf{x}_u = \mathbf{0}, 
 x_{u,\varnothing} = 1,\) and \(x_{u,\outsideOpt} = w_u/p_{\outsideOpt}\) is always a feasible solution. 

The goal of the SP is to set the prices such that the resulting allocation is a \emph{competitive equilibrium}:
\begin{definition}\label{def:CompetitiveEquilibrium}
\((\bar{\mathbf{x}}^\ast, \mathbf{p}^\ast)\) is said to be a \emph{competitive equilibrium} if the following conditions are satisfied 
\begin{itemize}
    \item[(i)] For every \(u\in U,\) \(\bar{\mathbf{x}}_u^\ast\) is an optimal solution of \eqref{eq: UserOpt} with prices set to \(\mathbf{p}^\ast\); 
    \item[(ii)] The capacity constraints are satisfied. That is, for every \(e\in \tilde{\mathcal{E}}\), \(
    \sum_{u\in U}x^{\ast}_{u,e} \leq \ell_e. 
\)
\item[(iii)] \(p^{\ast}_e \geq 0\) for all \(e\in \mc{\tilde{E}}\); and if \(p^{\ast}_e > 0 \) then \(\sum_{u\in U}x_{u,e}^{\ast}(p^\ast) = \ell_e\). 
\end{itemize} 
\end{definition}
We call \(\bar{\mathbf{x}}^\ast,\mathbf{p}^\ast\) to be the market clearing allocation and market clearing prices, respectively.
In general, a competitive equilibrium may not always exist \cite{7c65302b-f079-361a-94f1-0c3c9f6fc76b}. Therefore, we introduce a relaxed version of competitive equilibrium, where we relax the requirement that allocations are integral.
\begin{definition}
    \((\bar{\mathbf{x}}^\ast, \mathbf{p}^\ast)\) is called a \emph{fractional-competitive equilibrium} if all conditions in Definition \ref{def:CompetitiveEquilibrium} are satisfied, except that in  Definition \ref{def:CompetitiveEquilibrium}-(i) the integral constraint \eqref{eq: IntegralConstraints} is relaxed to a positivity constraint. 
\end{definition}
This relaxation is inspired by the competitive equilibrium framework studied in the literature on Fisher markets with linear constraints \cite{jalota2023fisher}. 

In what follows, we demonstrate that a fractional competitive equilibrium always exists and can be computed as the solution to a fixed-point problem. Next, in Section \ref{sec: Algorithm}, we leverage this property to develop a two-step algorithmic procedure that produces an integral allocation and prices approximating a competitive equilibrium, without requiring knowledge of the private valuations of the AAM vehicles.

\subsection{Existence and Computation of Fractional Competitive Equilibrium}
We state the following existence result about the fractional competitive equilibrium.
\begin{proposition}\label{prop: MarketClearingExists}
There exists a fractional-competitive equilibrium. 
\end{proposition}
The proof builds on the result establishing the existence of a competitive equilibrium in Fisher markets with auxiliary inequality constraints \cite{jalota2023fisher}. Specifically, our proof accounts for auxiliary \emph{equality} constraints that arise from \eqref{eq: SelectingRoute}-\eqref{eq: PathConstraints}. A detailed proof of Proposition \ref{prop: MarketClearingExists} is provided in Appendix \ref{sec:Proposition33}.

Next, we present a computational framework that can be used by the service provider for computing a fractional-competitive equilibrium, if they know the private valuations of all AAM vehicles. Consider the following optimization problem, parametrized by \(\omega\in\mathbb{R}_{\geq 0}^{|U|}:\)
\begin{subequations}
\begin{align}
       \underset{\bar{\mathbf{x}}=(\bar{\mathbf{x}}_u)_{u\in U}}{\max} &\quad  \sum_{u\in U} (w_u + \omega_u) \log\left(f_u(\bar{\mathbf{x}}_u)\right) - \sum_{u\in U} p_{\outsideOpt}x_{u,\outsideOpt}  \label{eq: Obj_SOPM} \\ 
    \text{s.t.} & \quad \sum_{u\in U} x_{u,e}\leq \ell_e \quad \forall \ e\in \tilde{\mathcal{E}} \label{eq: Cap_ConstraintM} \\
    {} & \quad \tilde{\mathbf{a}}_u^\top \mathbf{x}_u  + x_{u,\varnothing} = 1 \quad \forall \ u \in U \label{eq: SelectingRouteM}\\
{} & \quad \tilde{\mathbf{A}}_{u}\mathbf{x}_u = \mathbf{0} \quad \forall \ u \in U 
   \label{eq: PathConstraintsM} \\ 
    {} & \quad x_{u,e}\geq 0 \quad \forall \ u\in U, e\in \tilde{\mathcal{E}}\cup\{\outsideOpt, \dropOut\} \label{eq: Pos_constraintM},
\end{align}\label{eq: PlannerOptMain}
\end{subequations}
\noindent where the first term in \eqref{eq: Obj_SOPM} represents the ``budget-adjusted'' weighted geometric mean of the utilities of all AAM vehicles, while the second term accounts for the total expenditure on the outside option. Constraint \eqref{eq: Cap_ConstraintM} enforces the capacity limit on every edge, and constraints \eqref{eq: SelectingRouteM}-\eqref{eq: PathConstraintsM} are analogous to \eqref{eq: SelectingRoute}-\eqref{eq: PathConstraints}. Additionally, \eqref{eq: Pos_constraintM} is a relaxation of the integrality constraint in \eqref{eq: IntegralConstraints}. The objective in \eqref{eq: PlannerOptMain} is related to the ``budget-adjusted social optimization problem'' studied in \cite{jalota2023fisher} for Fisher markets, with the key difference being the second term, which ensures that a smaller amount of credits is spent on the outside option.  
On an intuitive level, the weighted geometric mean structure of the objective in \eqref{eq: Obj_SOPM} can be seen as finding an allocation that balances the trade-offs between different AAM vehicles’ utilities, weighted by their budget adjustments. It ensures that an improvement in a vehicle’s utility contributes to the overall objective in proportion to its market power. At the optimal point, this results in a fair and efficient allocation of airspace among AAM vehicles.
On a more technical level, the weighted geometric mean is fundamental in the proof of Proposition 4.5, which ensures that if we can adjust the weights of AAM vehicles in an appropriate manner (through a careful choice of \(\omega\)) then the optimal dual multiplier of (3b) is the market clearing price and the optimizer of (3) are market clearing allocations. 

Before stating Proposition \ref{prop: FixedPointOptimization}, we present an important property of \eqref{eq: PlannerOptMain}.
\begin{lemma}\label{lem: ConvexOptMain}
    The constraints \eqref{eq: Cap_ConstraintM}-\eqref{eq: Pos_constraintM} always have a feasible solution. Furthermore, for any \(\omega\in \mathbb{R}^{|U|}_{+}\), \eqref{eq: PlannerOptMain} is a convex optimization problem. 
\end{lemma}
The proof of this result follows from the fact that the constraint set in \eqref{eq: PlannerOptMain} is a polytope. Moreover, the objective \(f_u(\cdot)\) is a linear function with positive coefficients. 
For any \(\omega\in \mathbb{R}^{|U|}_{+}\), let \(\mathbf{p}^\dagger(\omega)\) denote an optimal dual multiplier corresponding to the constraint \eqref{eq: Cap_ConstraintM}, \(\lambda^\dagger(\omega)\) denote an optimal dual multiplier corresponding to the constraint \eqref{eq: SelectingRouteM}, and \(\bar{\mathbf{x}}^\dagger(\omega)\) denote an optimal solution to \eqref{eq: PlannerOptMain}.   
\begin{proposition}\label{prop: FixedPointOptimization}
    Suppose there exists \(\omega^\ast\in \mathbb{R}^{|U|}_{+}\) that is a fixed point of the mapping \(\omega \mapsto  \lambda^\dagger(\omega)\). Then \((\bar{\mathbf{x}}^\dagger(\omega^\ast), \mathbf{p}^\dagger(\omega^\ast))\) is a fractional-competitive equilibrium. 
    \end{proposition}
The proof relies on the convexity of the optimization problems \eqref{eq: PlannerOptMain} (Lemma \ref{lem: ConvexOptMain}) and \eqref{eq: UserOpt} (after relaxing the integrality constraint in \eqref{eq: IntegralConstraints} to fractional in \eqref{eq: Pos_constraintM}), along with matching the KKT conditions for optimality. The detailed proof of Proposition \ref{prop: FixedPointOptimization} is provided in Appendix \ref{sec:Proposition35}. 

\begin{remark}
   The proof of Proposition \ref{prop: FixedPointOptimization} can be extended to settings where a fixed point may not exist. Suppose there exists \(\omega^\ast \in \mathbb{R}^{|U|}_{+}\) such that, for each \(u \in U\), \(\omega^\ast_u - \lambda^\dagger_u(\omega^\ast) = \epsilon_u\) for some \(\epsilon_u\in \mathbb{R}\) ensuring that \(w_u+\epsilon_u\geq 0\). Then \((\bar{\mathbf{x}}^\dagger(\omega^\ast), \mathbf{p}^\dagger(\omega^\ast))\) is a fractional-competitive equilibrium of a market where, for each \(u \in U\), the budget is adjusted to \(w_u + \epsilon_u\).
\end{remark}

%% file: Sections/Algorithms.tex
In this section, we outline our algorithmic procedure for the service provider (SP) to compute (approximate) competitive equilibria, using Proposition \ref{prop: FixedPointOptimization}, without knowing the private valuations of AAM vehicles. Our approach unfolds in two stages. First, we introduce an algorithm that solves the fixed-point equation from Proposition \ref{prop: FixedPointOptimization} to compute the fractional-competitive equilibrium in a distributed manner (cf. Algorithm \ref{algo:1}). The SP then generates a ranked list of AAM vehicles using the prices and fractional allocations derived from this step. This ranking allows the SP to achieve an integral allocation by successively assigning regions to AAM vehicles according to the ranking (cf. Algorithm \ref{algo:2}). In the following subsections, we elaborate on each of these steps.

\subsection{{\textbf{Step 1:}} Distributed Algorithm for Computing Fractional-Competitive Equilibrium}
To compute a fractional-competitive equilibrium, we propose Algorithm \ref{algo:1} to solve the fixed-point problem described in Proposition \ref{prop: FixedPointOptimization} in a distributed manner. Algorithm \ref{algo:1} emulates the fixed-point iteration for the mapping
\begin{align}\label{eq: FixedPointError}\tag{\textsf{FP}}
    \omega \mapsto \lambda^{\dagger}(\omega).
\end{align}
Since the SP lacks access to \(\lambda^{\dagger}(\omega),\) as computing it requires solving \eqref{eq: PlannerOptMain}, which in turn depends on the private valuations of AAM vehicles, we adopt a two-loop approach to circumvent this challenge. In the inner loop, we iteratively solve the convex optimization problem \eqref{eq: PlannerOptMain} in a distributed manner that does not require the SP to access private valuations of AAM vehicles. This is done by repeatedly interacting with the AAM vehicles for a finite number of rounds, while holding \(\omega\) constant, to approximate \(\lambda^{\dagger}(\omega)\). This approximation is then used to update \(\omega\) using a fixed point iteration in the outer loop.

To solve the inner-loop problem in a distributed manner, we reformulate \eqref{eq: PlannerOptMain} into \eqref{eq: PlannerOptADMM} by introducing two additional variables, \({\mathbf{y}}\) and \({\mathbf{z}}\).  This reformulation enables the use of distributed optimization techniques, facilitating distributed computation across multiple agents while preserving the structure of the original problem.
\begin{subequations}
    \begin{align}
    \min_{(\bar{\mathbf{x}}_u, {\mathbf{y}}_u)_{u\in U}, ({z}_e)_{e\in \tilde{\mathcal{E}}}} &\quad  \sum_{u\in U}
(w_u + \omega_u) \log\left(f_u(\bar{\mathbf{x}}_u)\right) - \sum_{u\in U} p_{\outsideOpt}x_{u,\outsideOpt}  \label{eq: PlannerObjReformulate}\\ 
    \text{s.t.} & \quad \mathbf{y}_u = \mathbf{x}_u\quad \forall \ u\in U, \label{eq: PlannerObjReformulatex_y}
    \\ 
    {}& 
    \quad \sum_{u\in U} y_{u,e} + z_e = \ell_e \quad \forall \ e\in \tilde{\mathcal{E}} \label{eq: PlannerObjReformulate_Slack} \\
    {} & \quad \tilde{\mathbf{a}}_u^\top \mathbf{x}_u  + x_{u,\varnothing} = 1 \label{eq: SelectingRouteMRef}\\
{} & \quad \tilde{\mathbf{A}}_{u}\mathbf{x}_u = \mathbf{0}
   \label{eq: PathConstraintsMRef}  \\ 
    {} & \quad \bar{\mathbf{x}}_u\geq \mathbf{0}, {\mathbf{z}} \geq \mathbf{0} , \mathbf{y}_u\in\mathbb{R}^{|\tilde{\mathcal{E}}|}, \quad \forall \ u\in U. \label{eq: PlannerObjReformulate_PosX}
\end{align}\label{eq: PlannerOptADMM}
\end{subequations}

\noindent In this reformulation, Equation \eqref{eq: PlannerObjReformulatex_y} enforces the equality between \({\mathbf{x}}\) and  \({\mathbf{y}}\),  while Equation \eqref{eq: PlannerObjReformulate_Slack} ensures that capacity constraints are met. Constraints \eqref{eq: SelectingRouteMRef}-\eqref{eq: PlannerObjReformulate_PosX} is identical to \eqref{eq: SelectingRouteM}-\eqref{eq: Pos_constraintM} along with additional positivity constraints on \({\mathbf{y}}, {\mathbf{z}}.\) 
This reformulation ensures that the Lagrangian of \eqref{eq: PlannerOptADMM} becomes a separable function of \(\bar{\mathbf{x}}_u\), allowing the problem to be solved in a distributed manner using the ADMM algorithm \cite{he2023extensions, jalota2023fisher}. The variable \({\mathbf{y}}\) can be interpreted as the ``expected allocation'' estimate of the service provider (SP), while \({\mathbf{z}}\) represents the ``resource surplus'' in each region.
 Next, we describe the inner and outer loops in detail. We index the inner loop iterations by \(n = 1,2,...,N\) and the outer loop iterations by \(k=1, 2, ..., K\). A flowchart of our algorithm in Step 1 is shown in Figure \ref{fig:Picture Algorithm 1}.  

 \begin{figure}
     \centering
     \includegraphics[width=0.9\linewidth]{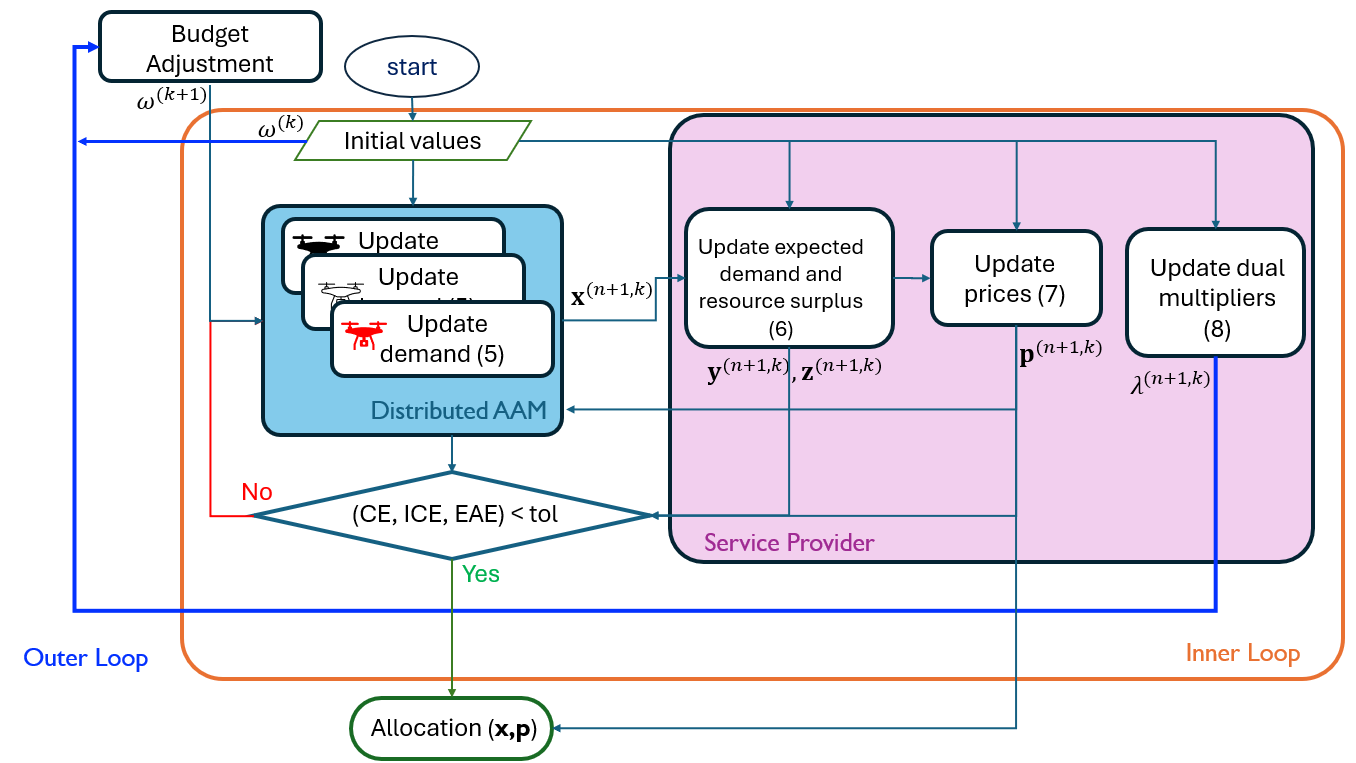}
     \caption{Flowchart of Algorithm \ref{algo:1}, illustrating the processes executed independently by the SP and AAM vehicles, as well as the steps computed within the inner and outer loops.}
     \label{fig:Picture Algorithm 1}
 \end{figure}
 
\begin{algorithm}
\fontsize{9pt}{12pt}\selectfont
    \caption{Distributed Algorithm for Fractional Competitive Equilibrium}
    \label{algo:1}
\begin{algorithmic}[1] 
    \State $\textbf{Input:}~ \goodpricevector^{(0,0)} = \mathbf{0}, \mathbf{\lambda}^{(0,0)}=\mathbf{0}, \mathbf{y}^{(0,0)} = \mathbf{0}, \omega^{(0)} = \mathbf{0}, \textsf{tol}_{\textsf{CE}}, \textsf{tol}_{\textsf{ICE}}, \textsf{tol}_{\textsf{EAE}}, \beta$
    \For{\(k = 0, 1, 2, ..., K-1\)}\Comment{\texttt{\color{magenta}Outer loop starts}}
        \For{\(n = 0,1,..., N-1\)}\Comment{\texttt{\color{magenta}Inner loop starts}}

        \Statex {\color{magenta}\texttt{Distributed Updates:}}
       
        \State Using \(\mathbf{p}^{(n,k)},\) \(\mathbf{y}_u^{(n,k)}, \lambda_u^{(n,k)},\) \(\omega_u^{(k)}\), each AAM vehicle independently updates its demand \(\bar{\mathbf{x}}_u^{(n+1,k)}\) as follows 
        \begin{equation}
\begin{aligned}\label{eq: X_update}
    \bar{\mathbf{x}}_u^{(n+1, k)} = &\argmax_{\bar{\mathbf{x}}_u, \ \text{s.t.}~\eqref{eq: PathConstraintsM}-\eqref{eq: Pos_constraintM}~\text{hold}}\Bigg((w_u+\omega_u^{(k)})\log\left(f_u(\bar{\mathbf{x}}_u)\right) - p_{\outsideOpt}x_{u,\outsideOpt}   -\sum_{e\in\tilde{\mathcal{E}}}p^{(n,k)}_ex_{u,e}- \lambda_u^{(n,k)}\cdot (\tilde{\mathbf{a}}_u^\top\mathbf{x}_u + x_{u,\varnothing} - 1)  \\ &\quad \hspace{2cm}-  \frac{\beta}{2} (\tilde{\mathbf{a}}_u^\top\mathbf{x}_u + x_{u,\varnothing} - 1)^2-\frac{\beta}{2}  \|\mathbf{y}_u^{(n,k)}-\mathbf{x}_u\|^2 
    \Bigg),
\end{aligned}
\end{equation}
        
 where \(\beta\) is a positive scalar that represents the step-size in the ADMM algorithm (cf. Appendix \ref{sec:ADMM}).
        \Statex {\color{magenta}\texttt{Updates by Service Provider (SP):}}
        \State Using \((\bar{\mathbf{x}}_u^{(n+1,k)})_{u\in U}\), SP updates the expected allocation \(\mathbf{y}^{(n+1,k)}\) and resource surplus \(\mathbf{z}^{(n+1,k)}\) as follows:
        \begin{equation}\label{eq: y_barz_Update}
\begin{aligned}
    &(\mathbf{y}^{(n+1,k)},{\mathbf{z}}^{(n+1,k)}) = \argmax_{{\mathbf{y}}\in \mathbb{R}^{U \times |\tilde{\mathcal{E}}|}, \ {\mathbf{z}} \in \mathbb{R}^{|\tilde{\mathcal{E}}|}_+} \Bigg( - \frac{\beta}{2} \sum_{u\in U}\|\mathbf{y}_{u}-\mathbf{x}_{u}^{(n+1,k)}\|^2 - \frac{\beta}{2} \|\sum_{u\in U}\mathbf{y}_u + \mathbf{z} - \ell\|^2 - \sum_{e\in \tilde{\mathcal{E}}} p^{(n,k)}_e {z}_e\Bigg).
\end{aligned}
\end{equation}
        
        \State Using \((\mathbf{y}^{(n+1,k)}, \mathbf{z}^{(n+1,k)}),\) SP updates the price estimates as follows
\begin{equation}\label{eq: price_Update}
\begin{aligned}
    \textbf{p}^{(n+1,k)} = \textbf{p}^{(n,k)} + \beta \left( \sum_{u\in U} \textbf{y}_u^{(n+1,k)} + {\mathbf{z}}^{(n+1,k)} - \ell \right).
\end{aligned}
\end{equation}
        
        \State
        For every \(u\in U,\) using \(\bar{\mathbf{x}}_u^{(n+1, k)},\) SP updates the dual multiplier \(\lambda_u^{(n+1,k)}\) as follows: 
\begin{equation}\label{eq: lambda_update}
\begin{aligned}
\lambda_u^{(n+1,k)} = {\lambda}_u^{(n,k)} + \beta \left(\tilde{\mathbf{a}}_u^\top\mathbf{x}_u^{(n+1,k)} + x_{u,\varnothing}^{(n+1,k)} - 1 \right).
\end{aligned}
\end{equation}

        \If{$\textsf{CE}(\bar{\mathbf{x}}^{(n+1,k)}, {\mathbf{p}}^{(n+1,k)}) \leq \textsf{tol}_{\textsf{CE}},  \textsf{ICE}(\bar{\mathbf{x}}^{(n+1,k)}) \leq \textsf{tol}_{\textsf{ICE}}, \textsf{EAE}(\mathbf{x}^{(n+1,k)}, \mathbf{y}^{(n+1,k)})\leq \textsf{tol}_{\textsf{EAE}}$}
        \State \textbf{Return} \(\bar{\mathbf{x}}^\dagger= \bar{\mathbf{x}}^{(n+1,k)}, \bar{\mathbf{p}}^\dagger = {\mathbf{p}}^{(n+1,k)}\)
        \EndIf
         \EndFor\Comment{\texttt{\color{magenta}Inner loop ends}}
         \State For every \(u\in U,\) SP updates the  \(\omega^{(k+1)} = \lambda^{(N,k)}\)
         \State \(\mathbf{p}^{(0,k+1)} \leftarrow \mathbf{p}^{(N,k)},  \mathbf{y}^{(0,k+1)} \leftarrow \mathbf{y}^{(N,k)}\)
    \EndFor\Comment{\texttt{\color{magenta}Outer loop ends}}
    \State \textbf{Return} \(\bar{\mathbf{x}}^\dagger  = \bar{\mathbf{x}}^{(N,K)}, \bar{\mathbf{y}}^\dagger = {\mathbf{p}}^{(N,K)}\)
\end{algorithmic}
\end{algorithm}

\textit{\textbf{Inner Loop:}} Inner loop iterations are obtained by performing ADMM iterations\footnote{Derivation of ADMM updates for \eqref{eq: PlannerOptADMM} is provided in Appendix \ref{sec:ADMM}.} for \eqref{eq: PlannerOptADMM} with step-size parameter \(\beta\). For any \(\omega\in \mathbb{R}_{\geq 0}^{|U|},\) this implementation allows us to estimate the dual multiplier \(\lambda^\dagger(\omega)\) in a distributed manner without requiring knowledge of the private valuations of AAM vehicles. Next, we describe these iterations in more detail.

Given that outer loop is at iteration \(k\), at any iteration \(n\) of the inner loop: \((a)\) every AAM vehicle \(u\) keeps track of its individual demand \(\bar{\mathbf{x}}^{(n,k)}_u\), and \((b)\) the SP keeps track of three quantities: an estimate of expected allocations \({\mathbf{y}}^{(n,k)}\), expected resource surplus \({\mathbf{z}}^{(n,k)}\), price of all regions \(\mathbf{p}^{(n,k)}\), and a dual multiplier \(\lambda^{(n,k)}\) which is used to adjust budgets of AAM vehicles. 

\underline{Local update for each AAM vehicle:}
Given that the outer loop is at iteration \(k\), at every iteration \(n\) of the inner loop, each AAM vehicle \(u\) receives its expected expected allocation \({\mathbf{y}}^{(n,k)}_u\), the current prices on regions \(\mathbf{p}^{(n,k)}\) and the dual multiplier corresponding to its local constraints \(\lambda^{(n,k)}_u\). Using this information, AAM vehicle updates its requested demand using \eqref{eq: X_update} and shares this with SP.  

\begin{remark}\label{rem: AAM vehicleUpdate}
    The update in Equation \eqref{eq: X_update} can be implemented through a ``proxy bidding agent'' in place of an actual AAM vehicle. The AAM vehicle operator can have the proxy agent at their local end where they feed AAM vehicle valuation and then the proxy agent participates on behalf of AAM vehicle, which ensures that no on-board energy resources are used to run Algorithm 1. This proxy attempts to maximize the vehicle's budget-adjusted utility while penalizing deviations from the expected allocation, overspending artificial currency, and violating constraints. These constraints capture the requirement that the bundle of edges selected by the AAM vehicle results in a feasible path on the time-extended graph. 
\end{remark}

\underline{Updates by SP:}
Using the demand from AAM vehicles, the SP updates the expected allocation and the excess supply through \eqref{eq: y_barz_Update}.  
The objective in \eqref{eq: y_barz_Update} requires that SP minimize three terms: \((i)\) difference between the expected allocation by SP and the demand sent by AAM vehicles; \((ii)\) violation of constraints in the resources; and \((iii)\) the unused capacity is minimized on any resource with a positive price. 

Next, the SP updates the price estimates using the updated values of the expected allocation and the excess supply through \eqref{eq: price_Update}. 
Equation \eqref{eq: price_Update} resembles the idea that an SP should increase the price if the capacity constraint is violated and reduce it if there is available capacity. 

Finally, the SP updates the dual multiplier estimate \(\lambda_u,\) for every AAM vehicle \(u\) using \eqref{eq: lambda_update}. 

\textbf{\textit{Outer loop.}} In the outer loop, we update the budget adjustment after every \(N\) iteration of the inner loop, using the value of \(\lambda^{(N,k)}\) to approximate the fixed-point iteration (line 12 of Algorithm \ref{algo:1}).  This step ensures that budget adjustments progressively converge toward equilibrium by iteratively refining the dual variables based on the current solution from the inner loop.

\textbf{Termination criterion:}
We terminate the algorithm once \textbf{all} of the following errors fall below their predefined threshold: 
\begin{itemize}
    \item Complementarity error (\textsf{CE}): Smaller values of \textsf{CE} ensure that resources with a positive price maintain a balance between demand and supply, while resources priced at zero satisfy capacity constraints. We define \begin{align}\label{eq: MCE}
\textsf{CE}(\bar{\mathbf{x}}, \mathbf{p})= \sqrt{\sum_{e\in \tilde{\mathcal{E}}} p_e^2 z_e^2},        
    \end{align}
where \({z}_e = \sum_{u\in U}x_{u,e} - \ell_e\) is the \emph{excess demand}. 
 The definition of \textsf{CE} is motivated from the ``complementarity condition'' in general equilibrium theory in economics \cite{mas2006microeconomic}. We define \(\textsf{tol}_{\textsf{CE}}\) to be the threshold for this error.

\item Individual constraint error (\textsf{ICE}): Smaller values of \textsf{ICE} ensure that \eqref{eq: SelectingRouteMRef} constraint is satisfied. We define 
\begin{align}\label{eq: ICE}
    \textsf{ICE}(\bar{\mathbf{x}})= {\max_{u\in U} \|\tilde{\mathbf{a}}_u^\top\mathbf{x}_u + x_{u,\varnothing} - 1\|_{\infty}}.
\end{align}
 We define \(\textsf{tol}_{\textsf{ICE}}\) to be the threshold for this error.
\item Expected allocation error (\textsf{EAE}): Smaller value of \textsf{EAE} ensure that \eqref{eq: PlannerObjReformulatex_y} is satisfied. We define  
\begin{align}\label{eq: BDE}
    \textsf{EAE}(\mathbf{x}, \mathbf{y})= \max_{u\in U} \|\mathbf{y}_u-\mathbf{x}_u\|_{\infty}.
\end{align}
We define \(\textsf{tol}_{\textsf{EAE}}\) to be the threshold for this error.
\end{itemize}

We represent the output of Algorithm \ref{algo:1} by (\(\bar{\mathbf{x}}^{\dagger}, {\mathbf{p}}^{\dagger})\). 

\begin{remark}
Algorithm \ref{algo:1} bears resemblance to online learning algorithms for computing market equilibrium \cite{kleinberg2003value, li2020simple, bistritz2021online}. In this process, AAM vehicles iteratively adjust their demand based on anonymized signals—such as prices, expected allocation, resource surplus, and dual multipliers—broadcasted by the SP. Meanwhile, the SP dynamically updates these signals in response to the aggregate demand from AAM vehicles, without requiring knowledge of their private valuations.
\end{remark}

\subsection{{\textbf{Step 2:}} Computing Integral Allocation}
\begin{figure} 
    \centering
    \includegraphics[width=0.8\textwidth]{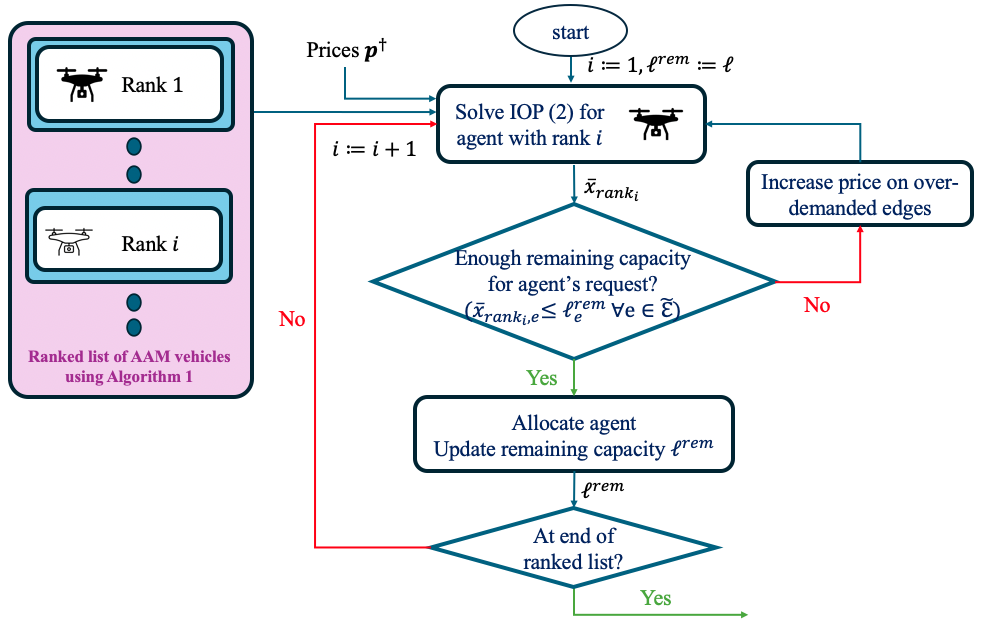}
    \caption{Flowchart of Algorithm \ref{algo:2}, illustrating the ranking system and the removal of over-demanded edges.}
    \label{fig: SchematicAlgorithm2}
\end{figure}

Using the output from Algorithm \ref{algo:1}, the service provider (SP) computes an integral allocation in a distributed manner using Algorithm \ref{algo:2}. 
The SP sets the airspace price to \(\mathbf{p}^\dagger\) (the output of Algorithm \ref{algo:1}) and generates a \textit{ranked list} of AAM vehicles based on \(\bar{\mathbf{x}}^\dagger\) (cf. Sec. \ref{sssec:RankedList}). This list is created by ranking the AAM vehicles in descending order according to the numerical value of the (fractional) allocation of their preferred routes (cf. Sec. \ref{sssec:IntegralAllocation}).

\subsubsection{Ranked List.}\label{sssec:RankedList} To generate a ranked list of AAM vehicles, we define \(s^\ast(u)\) to be the most desired route of AAM vehicle \(u\) in \(R_u\). Using \(\bar{\mathbf{x}}^\dagger\) from Algorithm \ref{algo:1}, the SP creates a ranking over agents based on decreasing values of \(x^\dagger_{u,e^\ast(s^\ast(u))}\), where \(e^\ast(s^\ast(u))\) denotes the departing edge on the route \(s^\ast(u)\). We denote the ranked list\footnote{Ties are broken arbitrarily.} of AAM vehicles by \(\textsf{\textbf{rank}}\).

\subsubsection{Integral Allocation.}\label{sssec:IntegralAllocation} 
After generating the ranked list, the SP fixes the prices for all resources based on the output of Algorithm \ref{algo:1} (i.e. $\bar{\mathbf{p}}^\dagger $ ) and iterates over AAM vehicles according to \textbf{\textsf{rank}} (cf. Line 3  in Algorithm \ref{algo:2}). Each AAM vehicle is allocated its most desired feasible route (cf. Line 5-10  in Algorithm \ref{algo:2}), subject to the current resource capacity. Suppose a capacity constraint is violated on any resource. In that case, the SP either removes that resource from the available pool for all remaining agents or increases its price to infinity for the remaining AAM vehicles (as described in line 12 of Algorithm \ref{algo:2}). See Fig. \ref{fig: SchematicAlgorithm2} for a schematic illustration.

\begin{algorithm}
    \caption{Integral Allocation Based on Ranked List}
\begin{algorithmic}[1]
    \State \textbf{Input:}~ \(\mathbf{p}^\dagger\),  \textbf{\textsf{rank}}  
      \State Initialize remaining capacity: \( \ell^{\text{rem}} \gets \ell \)
    \For{$i=1$ \textbf{to} $|U|$}
    \State \(u\leftarrow \textsf{\textbf{rank}}_i\)\Comment{\texttt{\color{magenta} Select AAM vehicle from the ranked list}}
    \State $\texttt{AAM\_vehicle\_allocated} = False $
    \While{$\texttt{AAM\_vehicle\_allocated} = False$}
        \State AAM vehicle \(u\) reports its \textbf{integral} allocation $\mathbf{\bar{x}}_{u}$ by solving \eqref{eq: UserOpt}
        \If {$x_{u,e} \leq \ell^{\text{rem}}_e$ for every \(e\in \tilde{\mathcal{E}}\)}
            \State Update remaining capacity: $\ell^{\text{rem}}_e \gets \ell^{\text{rem}}_e - \bar{x}_{u,e}, \quad \forall e \in \tilde{\mathcal{E}}$
        \State $\texttt{AAM\_vehicle\_allocated} = True $
        \Else
        \State Define contested set \(C \gets \{e \in \tilde{\mathcal{E}} \mid x_{u,e} > \ell^{\text{rem}}_e\}\) \Comment{\texttt{{\color{magenta}Identify contested good}}}
        \For{each \( e \in C \)}
            \State Update price: \( p_e \gets \infty \) \Comment{\texttt{{\color{magenta} Remove $e$ from further consideration}}}
        \EndFor
        \EndIf
    \EndWhile
    \EndFor
\State \textbf{Return} \( \bar{\mathbf{x}} \)
\end{algorithmic}   
    \label{algo:2}
\end{algorithm}
\begin{remark}\label{rem: AAM vehicle_Flight}
For Algorithms \ref{algo:1}-\ref{algo:2}, the SP only requires information on the resource demands of AAM vehicles, their feasible time trajectories, and their most preferred time trajectory. Importantly, the SP does \textbf{not} need any details regarding the private valuations of individual AAM vehicles. Moreover, each AAM vehicle does not have access to the demands or private valuations of other vehicles, ensuring that privacy is maintained throughout the process.
\end{remark}

%% file: Sections/Numerics.tex
In this section, we validate our proposed mechanism using a dataset of simulated package deliveries by Airbus, as shown in Fig. \ref{fig:cityToulouse}. 
 Particularly, we use a synthetic dataset generated by Airbus to simulate a drone-based package delivery scenario in Toulouse, France \cite{egorov2019encounter, chin2023protocol}. 

\textbf{Dataset Specification:} The data involves four warehouses located on the periphery of the city, which serve as hubs for UAV take-off and landing. Delivery requests are generated with spatial locations drawn uniformly across Toulouse and their temporal occurrences follow a Poisson process. 
Each request triggers a UAV to depart from the launchpad of a warehouse, deliver a package to the specified destination, and return to its origin. 
The UAV flight trajectory, including take-off, cruise, and landing phases, is generated by Airbus's high-fidelity trajectory simulator, ensuring that the simulated operations closely mimic real-world conditions. The data set includes data corresponding to $177$ UAV flights, which spans roughly 6000 seconds (\(100\) minutes). The average length of a UAV flight from a warehouse to the delivery location is approximately 300 seconds ($5$ minutes). 

\begin{figure}[h!]
    \centering
    \includegraphics[width=0.6\linewidth]{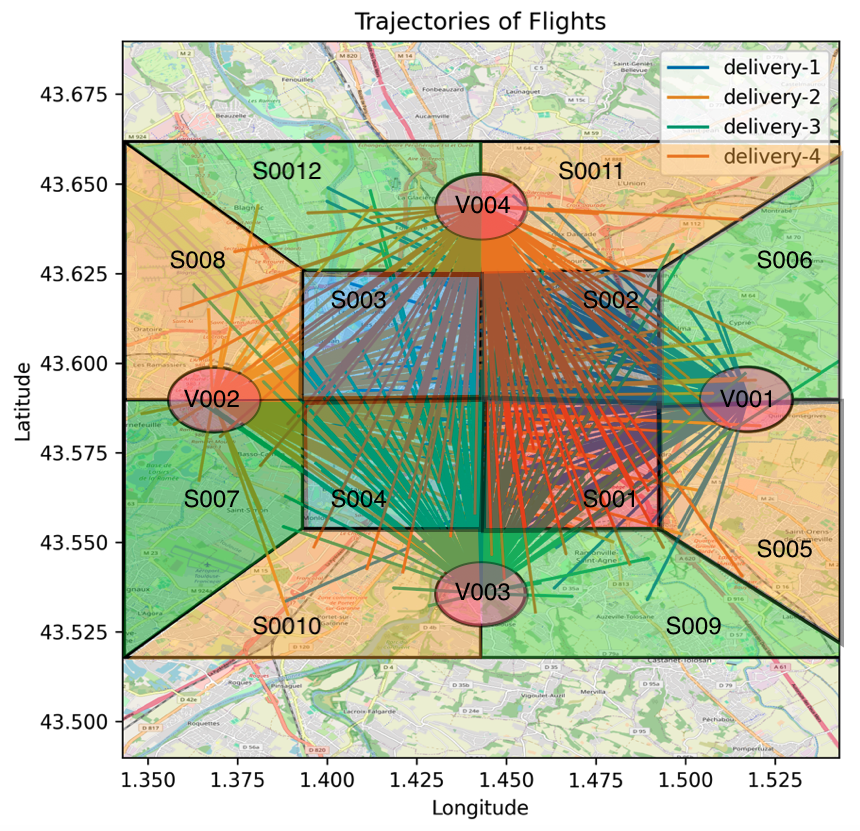}
    \caption{Division of Toulouse airspace into 12 cruising sectors (polygons) and 4 launch-pad sectors (circles). The lines show the trajectory of UAVs in the dataset. Labels indicate the sector (S\#) or vertiport (V\#).}
    \label{fig:cityToulouse}
\end{figure}

\textbf{Airspace Specifications:} 
To implement our auction approach, we partition  Toulouse's airspace into $12$ ``cruising-altitude'' regions, along with $4$ warehouses, as shown in Fig. \ref{fig:cityToulouse}. 
We construct a time-extended graph (cf. Definition \ref{def: Time-extendedGraph}) with a total of \(T = 400\) time-steps, where each step of corresponds to 15 seconds. Based on the data we find that the minimum capacity needed to accommodate all requests from UAVs is 14 units in all 12 cruising regions of the airspace and 4 units for vertiport departure and arrival at the warehouse locations. Therefore, to make this problem interesting, we set capacity to \(50\%\) of this in all regions, unless otherwise specified.

\textbf{UAV Specifications:} Each UAV is either allocated a path on the time-extended graph or is ``rebased.'' For every UAV, the feasible set of paths on the time-extended graph includes its most desirable path, along with four alternative paths that each incur a one-time-step delay. If a UAV is rebased, then it converts its budget into outside options to be used again in the next auction window. Every rebased UAV requests access to the airspace from the start of the next auction window. 
Any UAV may be rebased at most twice, after which it is dropped and not considered in future auctions.

The private value generated by each UAV on its most desirable path is a uniformly random number between \(150\) and \(250\) units. The utility decreases by a factor of \(0.95\) for each time-step delay in departure. If a UAV is rebased, its utility decreases by a factor of \(0.5\). The utility derived by any UAV from dropping out is \(40\) units.

\textbf{Implementation Specifications:} In each auction window, the SP allocates an additional budget of artificial currency to each participating UAV, randomly sampled between 150 and 250 units. This amount is then added to the UAV's existing budget accumulated from past auctions.   
For our numerical study, we set the nominal tolerance values in Algorithm \ref{algo:1} as follows:  
\(\textsf{tol}_{\textsf{CE}}\) is set to \(0.1\%\) of its maximum possible value,  
\(\textsf{tol}_{\textsf{ICE}}\) to \(0.01\%\) of its maximum value, and  
\(\textsf{tol}_{\textsf{EAE}}\) to \(0.1\%\) of its maximum value.  
Additionally, we set \(\beta = 50\) and \(N=30\) in Algorithm \ref{algo:1}.
We implemented our approach in Python and ran the simulations on a laptop equipped with a 12th-gen Intel Core i7-1200H CPU (14 cores, 20 threads) and 32GB DDR4 RAM. The operating system used was Ubuntu 22.04. Our code is publicly available at \href{https://github.com/sastry-group/Mechanism-Design-for-AAM}{https://github.com/sastry-group/Mechanism-Design-for-AAM}.

\subsection{Qualitative Analysis}
In this subsection, we present the outcome of our receding-horizon auction approach that we conducted for a total of \(13\) (i.e, \(I = 13\)) auctions.
\textcolor{black}{
Before the start of every auction, the SP gathers the demand of AAM vehicles requesting access to the airspace.
The overall budget of any UAV includes the new air credits they received and any unused air credits from a previous round.} Additionally, we also compare our performance with two baselines. 

In Fig. \ref{fig:drop_v_auction}, we present the number of agents that were allocated, delayed, dropped, rebased-once and rebased-twice in different auctions. We see that the number of rebased and dropped agents increases as we proceed to later rounds. This is because the congestion builds up as we are operating in highly contested settings, as we set the maximum capacity of every region to be at \(50\%.\) In Fig. \ref{fig:MarketClearningPercent}, we present the market clearing error (MCE) after Algorithm \ref{algo:2} which captures the fraction of edges on the time-extended graphs which have positive price but for which the congestion is below the capacity. {This metric aligns with the third component of the competitive equilibrium definition in Def. \ref{def:CompetitiveEquilibrium}}, capturing a common economic notion that we should not charge a payment on a good that is below its capacity. We see that this is a very small number, in the range of 0-0.6\%, highlighting that our approach (Algorithm \ref{algo:1} + \ref{algo:2}) is not imposing prices on the uncontested goods.

\begin{figure}
    \centering

    \begin{subfigure}{0.49\textwidth}
    \centering
    \includegraphics[width=\linewidth]{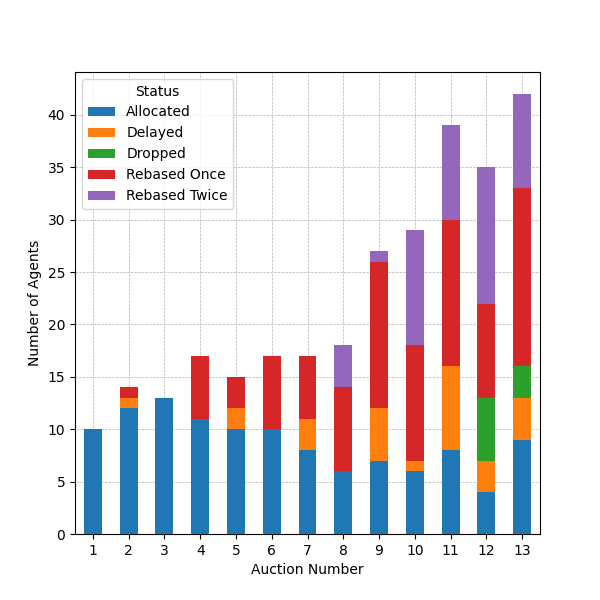}
    \caption{Allocations for Each Auction}
    \label{fig:drop_v_auction}
    \end{subfigure}
    \hfill
    \begin{subfigure}{0.49\textwidth}
        \centering
        \includegraphics[width=\linewidth]{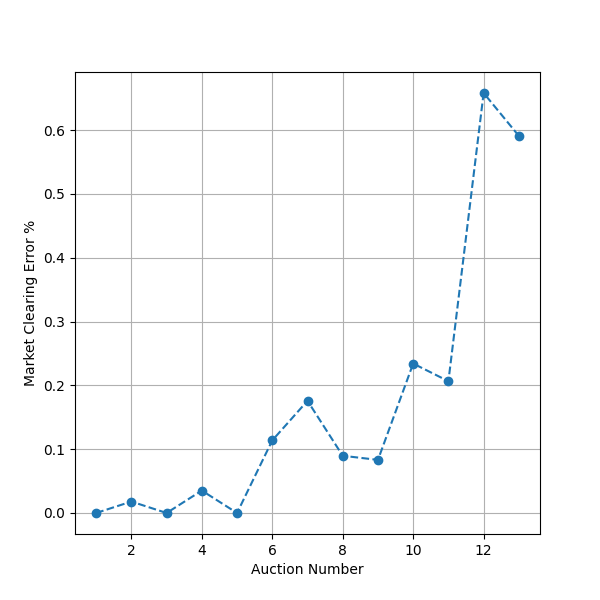}
        \caption{Market Clearing Error After Integral Allocation (percent)}
        \label{fig:MarketClearningPercent}
    \end{subfigure}

    \caption{ Properties of allocation finalized by our receding-horizon auction approach.}
    \label{fig:main}
\end{figure}

%% file: Sections/SupplementaryResultsMain.tex
\begin{table}[h!]
\centering
\footnotesize
\caption{Comparison to baseline auction approaches under different capacities}
\renewcommand{\arraystretch}{1.0}
\setlength{\tabcolsep}{2.5pt} 
\resizebox{\textwidth}{!}{ 
\begin{tabular}{|c|c|c|c|c|c|c|c|c|}
\hline
\textbf{Approach} & \textbf{Capacity} & \thead{\textbf{Num. Times} \\  \textbf{Rebased}} & \thead{\textbf{Num. Delayed} \\  \textbf{UAVs}}  & \thead{\textbf{Avg. Delay} \\ \textbf{(time steps)}} & \thead{\textbf{Num. Rebased} \\ \textbf{UAVs}} & \thead{\textbf{Avg. Times} \\ \textbf{Rebased}} & \thead{\textbf{Num. Never} \\ \textbf{Allocated}} \\ 
\hline
\multirow{2}{*}{\textbf{Budget-based}} & 60\% & 20 & 12 & \textcolor{ForestGreen}{\textbf{1.83}} & 17 & 1.18 & \textcolor{ForestGreen}{\textbf{0}} \\  
& 50\% & 148 & \textcolor{ForestGreen}{\textbf{22}} & \textcolor{ForestGreen}{\textbf{1.22}} & \textcolor{ForestGreen}{\textbf{84}} & 1.76 & 43 \\ 
\hline
\multirow{2}{*}{\textbf{Profit-based}} & 60\% & 59 & 17 & 2.41 & 32 & 1.84 & 24 \\  
& 50\% & 164 & \textcolor{ForestGreen}{\textbf{22}} & 2.86 & 87 & 1.89 & 70 \\ 
\hline
\multirow{2}{*}{\textbf{Ours}} & 60\% & \textcolor{ForestGreen}{\textbf{12}} & \textcolor{ForestGreen}{\textbf{10}} & 2.3 & \textcolor{ForestGreen}{\textbf{12}} & \textcolor{ForestGreen}{\textbf{1}} & \textcolor{ForestGreen}{\textbf{0}} \\  
& 50\% & \textcolor{ForestGreen}{\textbf{143}} & 27 & 1.48 & 96 & \textcolor{ForestGreen}{\textbf{1.49}} & \textcolor{ForestGreen}{\textbf{35}} \\  
\hline
\end{tabular}
}
\label{tab:merged_comparisons}
\end{table}

In Table \ref{tab:merged_comparisons}, we compare our approach to two baselines based on the ascending clock auction \cite{cramton2006combinatorial}. Both comparisons run simultaneous ascending clock auctions using $\beta$ as the price increment. Agents are allowed to perform price discovery by only bidding on their most beneficial request instead of all goods across their preferred and delayed requests. Due to (\ref{eq: SelectingRoute}), we can assume agents will always bid on a request or the outside option and are therefore always active. The procedure is outlined in Algorithm \ref{algo:comparisons}. In the "Budget-based" comparison, agents solve their IOP to determine their bid, and in the "Profit-based" they determine their bid based on their profit (value - price). In each round, all agents bid, and prices are raised on contested goods until no more goods are contested.

\begin{algorithm}
    \caption{Ascending Clock Auction Comparisons}
    \begin{algorithmic}[1]
    \State \textbf{Initialize} prices: \( \mathbf{p} \gets \mathbf{0} \)
    \Repeat
        \State $\mathbf{\hat{x}}_{u} \gets$ AAM vehicle \(u\) integral output from \eqref{eq: UserOpt} OR \Comment{\texttt{{\color{magenta}Budget-based Bid}}}
        \State $\mathbf{\hat{x}}_{u} \gets \text{argmax}_{\mathbf{\bar{x}}_u} \sum_{s\in R_u}v_{u,s} x_{u, e^\ast(s)} + v_{u,o}x_{u,o} - \textbf{p}^T\mathbf{x}_u - p_ox_{u,o} \quad \text{st.} \ (\ref{eq: SelectingRoute}) - (\ref{eq: IntegralConstraints}) \quad \forall u \in U$ \Comment{\texttt{{\color{magenta}Profit-based Bid}}}
        \State $\mathcal{C} = \{ e \in \
        \mathcal{\tilde{E}} \ | \ \ell_e < \Sigma_{u \in U} \hat{x}_{u,e} \}$
        \State $p_e \gets p_e + \beta \quad \forall e \in \mathcal{C}$ \Comment{\texttt{{\color{magenta}Increase Price for Overcapacitated Goods}}}
    \Until{$\mathcal{C} = \emptyset$}
    \State \textbf{Return} \( \hat{\mathbf{x}}_{u \in U} \)
    \end{algorithmic}   
    \label{algo:comparisons}
\end{algorithm}

We compare these approaches using the Toulouse example at a 50\% capacity level and also consider a slightly less constrained case (60\% capacity). The set of goods remains the same across all approaches. Our evaluation includes the number of agents that are never allocated (dropped agents), the number of delayed agents, the average delay duration for delayed agents, the total number of times agents are rebased (including agents that are rebased-once and rebased-twice), the number of agents that are rebased at least once, and the average number of times agents are rebased. Lower values are preferable across all metrics, and the lowest value for each case is shown in green and bolded.

For the 60\% capacity case, both our approach and the Budget-based comparison allocate all agents, with our approach additionally resulting in fewer delayed and rebased aircraft. In the more constrained 50\% case, our approach significantly reduces the number of unallocated agents. While the Budget-based approach results in fewer delayed aircraft and a lower number of rebased agents, minimizing the number of unallocated agents is the more desirable outcome. The Profit-based approach performs worse in both cases across almost every metric, further highlighting the benefits of using artificial currency. We attribute this to the high cumulative costs incurred when agents must bid on multiple goods in this setting. These results show that, for comparable betas,  our approach can better handle rising congestion due to the lower numbers of unallocated and rebased agents.

\subsection{Sensitivity Analysis}

\begin{figure}
    \centering
    \begin{subfigure}{0.45\textwidth}
        \centering
        \includegraphics[width=\linewidth]{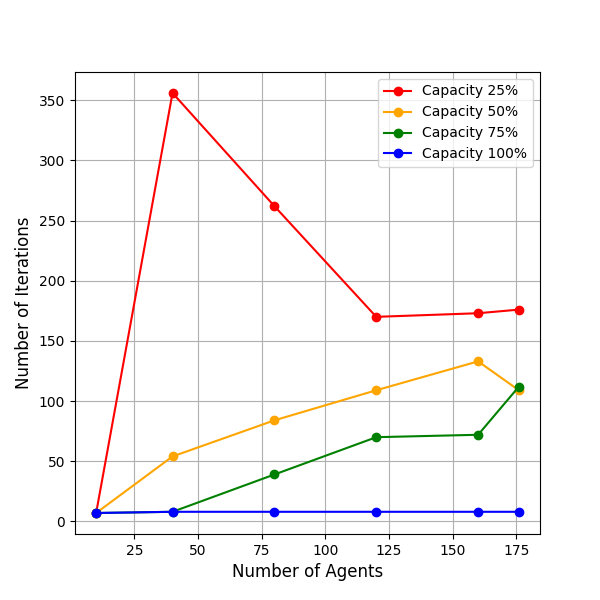}
        \caption{Iterations vs. Number of UAVs}
        \label{fig:iter_v_agents}
    \end{subfigure}
    \hfill
    \begin{subfigure}{0.45\textwidth}
    \centering
    \includegraphics[width=\linewidth]{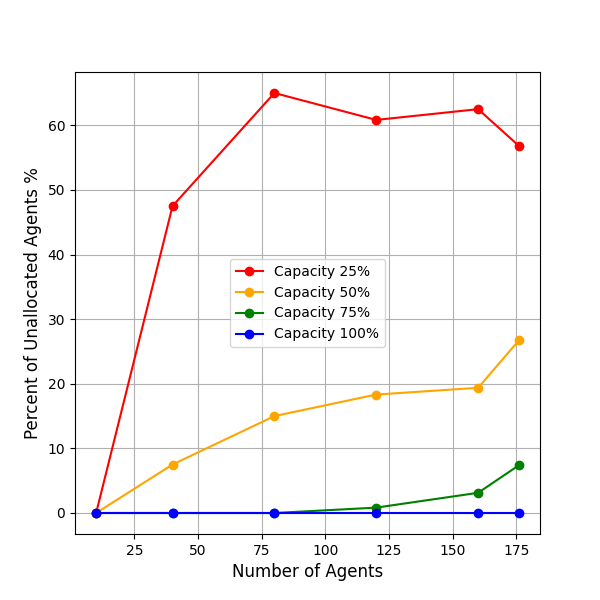}
    \caption{Percent of Unallocated UAVs vs. Number of UAVs}
    \label{fig:Percent_Unallocated}
    \end{subfigure}
    
    \caption{Variation of the number of iterations and percent of unallocated agents with respect to the number of agents as we vary capacity constraints on the resources.}
    \label{fig:performance_num_agents} 
\end{figure}

\begin{figure}
    \centering
    \begin{subfigure}{0.45\textwidth}
        \centering
        \includegraphics[width=\linewidth]{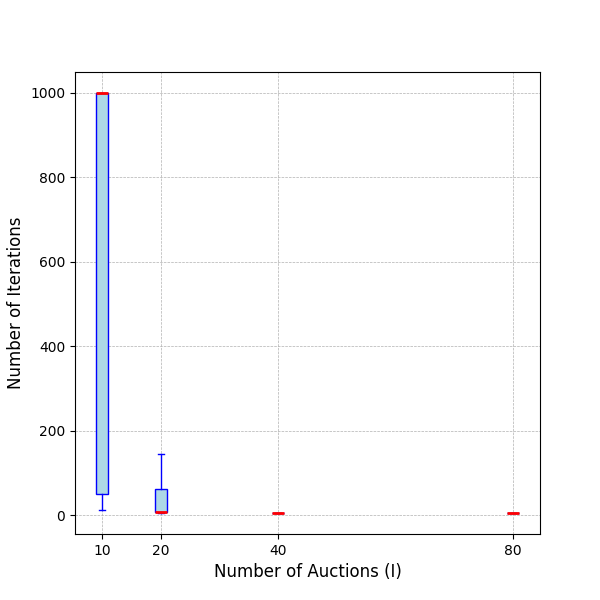}
        \caption{Iterations vs. Number of Auctions (I)}
        \label{fig:num_auction}
    \end{subfigure}
    \hfill     
    \begin{subfigure}{0.45\textwidth}
        \centering
        \includegraphics[width=\linewidth]{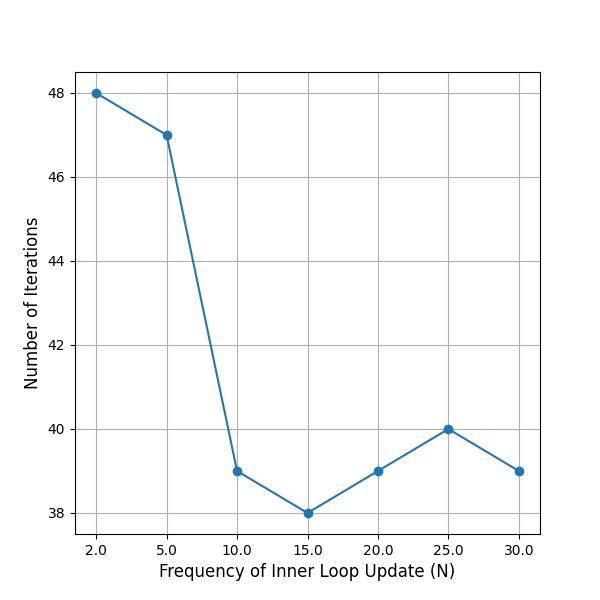}
        \caption{Iterations vs. Inner Loop Update Parameter in Algorithm \ref{algo:1} (N)}
        \label{fig:lamda_study}
    \end{subfigure}

    \begin{subfigure}{0.4\textwidth}
        \centering
        \includegraphics[width=\linewidth]{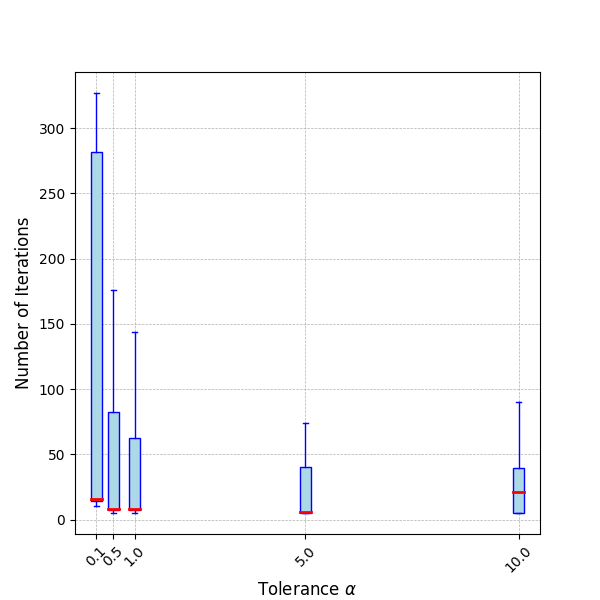}
        \caption{Iterations vs. Tolerance in Algorithm \ref{algo:1}}
        \label{fig:iter_vs_tol_whisker}
    \end{subfigure}

    \caption{ Variation in number of iterations of Algorithm \ref{algo:1} with different algorithmic parameters}
    \label{fig:performance}
\end{figure}

In this subsection, we conduct numerical studies to evaluate the impact of key design parameters on the performance of our algorithmic approach.

In Fig. \ref{fig:iter_v_agents}, we present the variation in the number of iterations of Algorithm \ref{algo:1} with respect to different numbers of agents under various capacity constraints. This scenario can be understood as a setting where there is only one auction window in which all agents simultaneously request their desired goods on the time-extended graph.
We observe that an increase in the number of agents leads to a higher number of contested goods, which requires more iterations of Algorithm \ref{algo:1} to compute a fractional competitive equilibrium. At 100\% capacity, there are no contested resources, and the algorithm requires only 8 iterations regardless of the number of agents. As the capacity decreases, the number of iterations increases due to a higher number of contested goods. 
In Fig. \ref{fig:Percent_Unallocated}, we analyze the variation in the percentage of unallocated agents with respect to different numbers of agents under various capacity constraints. We observe that when the resource capacity is 100\%, all UAVs receive their desired paths. However, as the capacity decreases, a greater number of agents remain unallocated. Furthermore, for any given capacity level, the percentage of unallocated agents increases as more agents participate in the auction.

In Fig. \ref{fig:performance}, we study the variation in the number of iterations of Algorithm \ref{algo:1} with respect to different market parameters.
First, in Fig. \ref{fig:num_auction}, we examine how the number of iterations of Algorithm \ref{algo:1} changes as we vary the number of auctions in our receding horizon approach. As the number of auctions increases, the number of UAVs participating in each auction decreases, resulting in fewer contested goods. Consequently, Algorithm \ref{algo:1} converges more quickly.
Next, in Fig. \ref{fig:lamda_study}, we evaluate the impact of the number of inner loop updates (parameter \(N\) in Algorithm \ref{algo:1}). Recall, that the goal of Algorithm \ref{algo:1} is to emulate fixed point iteration (\ref{eq: FixedPointError}). Towards this goal, the role of inner loop updates in Algorithm \ref{algo:1} is to estimate $\lambda^{\dagger}(\omega^{(k)})$ for every update of $\omega^{(k)}$ in the outer loop.  The results show that if \(N\) is lower, the number of iterations of Algorithm \ref{algo:1} is higher, as the inner loop cannot estimate \(\lambda^\dagger(\omega^{(k)})\) accurately. Consequently, increasing \(N\) decreases the number of iterations up to a point. However, if we continue to increase \(N\) past this point then the additional steps in the inner loop do not help in the convergence of fixed point iteration (\ref{eq: FixedPointError}) and instead cause the iterations of Algorithm \ref{algo:1} to start increasing again. 
Finally, in Fig. \ref{fig:iter_vs_tol_whisker}, we analyze how the number of iterations varies with different convergence tolerances. For clear presentation, we introduced the variable \(\alpha\), which is a multiplier for the nominal value of tolerances (i.e., \(\textsf{tol}_{\textsf{CE}}, \textsf{tol}_{\textsf{ICE}}, \textsf{tol}_{\textsf{EAE}}\)) described in our setup above and we study the variation in the number of iterations with respect of \(\alpha.\) As expected, stricter tolerances (i.e., lower \(\alpha\)) generally increase the number of iterations, since the algorithm must satisfy more stringent stopping conditions.

%% file: Sections/Limitations.tex
Here, we discuss some limitations of the current modeling framework. Addressing these limitations presents interesting directions for future research.
\begin{itemize}
    \item[(i)] We assume that the AAM vehicles are not malicious and they act according to Algorithms \ref{algo:1} and \ref{algo:2}. One way to implement our mechanisms is that before takeoff, each AAM vehicle interacts with the SP through a ``proxy agent'', which is an autonomous entity that bids on its behalf based on Algorithm \ref{algo:1} and \ref{algo:2}. The use of proxy agents is a popular method for implementing iterative auctions \cite{cramton2006combinatorial}. Developing suitable auditing mechanisms to ensure that agents act honestly is an interesting open problem. 
    \item[(ii)] We consider that the communication between the service provider and each UAV is delay-free and noiseless. Once the vehicle has taken off, its route is fixed, and no further in-flight communication is required.
    \item[(iii)] We assume that a single service provider controls the airspace. An interesting direction for future research is to extend our approach when the airspace is managed by multiple service providers. 
    \item[(iv)] In the current framework, each UAV is allocated to a feasible trajectory before departing and they do not change the trajectory during flight. 
    \item[(v)] We assume that adequate infrastructure and operational measures are in place to handle emergencies.
    \item[(vi)] We consider that the SP randomly assigns budget to all AAM vehicles, but developing better budget allocation schemes that result in socially optimal outcome is an interesting direction of future research. 
\end{itemize}

%% file: Sections/Conclusion.tex
In this work, we introduce a novel mechanism that enables service providers to allocate on-demand requests from Advanced Air Mobility (AAM) vehicles, each with heterogeneous private valuations, to a capacity-constrained airspace. We evaluated the effectiveness of our approach using a physically accurate urban air delivery dataset.
This is the first work in the AAM literature to allocate constrained airspace resources to dynamically arriving AAM vehicles without requiring knowledge of their private valuations. Core to our approach is an artificial currency-based auction mechanism that is implemented in a receding-horizon manner, wherein for each auction, we use a distributed iterative algorithm that accounts for individual agent preferences while ensuring system efficiency and safety.

%% file: Sections/TimeExtendedGraph.tex
In this section, we explain the time-extended graph (Definition \ref{def: Time-extendedGraph}) along with constraints \eqref{eq: SelectingRoute}-\eqref{eq: PathConstraints} through a simple example comprising of 3 regions, denoted by  \(\{A,B,C\}\). 
\begin{figure}
    \centering
    \includegraphics[width=\textwidth]{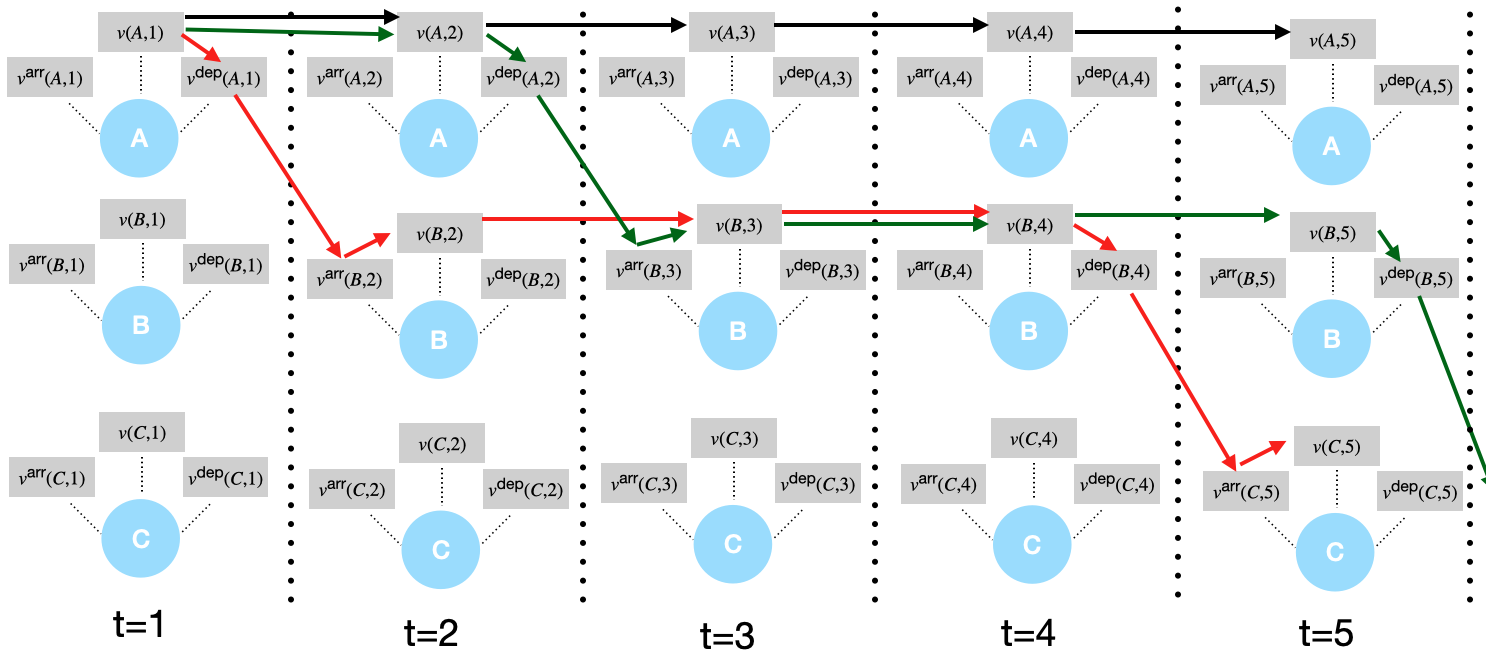}
    \caption{Time Extended Graph: From left to right, we show a sequence of time steps and different color-coded trajectories that an AAM vehicle can request. The red trajectory shows an AAM vehicle traveling from region A to C while transiting from region B. The green trajectory represents the same trajectory as the red path but is delayed by one unit of time. The black trajectory denotes an option where the AAM vehicle stays parked at the origin region. To simplify the visualization, we have not shown all possible edges on this time-extended graph.}   
    \label{fig:sector-price}
\end{figure}

\paragraph{Time-extended graph}
The time-extended graph  (for \(T=5\)) corresponding to our scenario is shown in Figure \ref{fig:sector-price}. 
In the time-extended graph $\tilde{\mc{G}} = (\tilde{\mc{R}}, \tilde{\mc{E}})$,  at every time \(t\) each region $r$ is replicated into three regions $t:v(r,t), v^\arr(r,t), v^\dep(r,t)$. For conciseness, we will only discuss one edge corresponding to each of the four types \(\tilde{\mathcal{E}}^{(1)}, \tilde{\mathcal{E}}^{(2)}, \tilde{\mathcal{E}}^{(3)}, \tilde{\mathcal{E}}^{(4)}\). The (red) edge \((v^{\arr}(B,2), v(B,2))\) is an edge of type \(\tilde{\mathcal{E}}^{(1)}\). The (red) edge \((v(A,1), v^{\dep}(A,1))\) is an edge of type \(\tilde{\mathcal{E}}^{(2)}\). The (black) edge \((v(A,1), v(A,2))\) is an edge of type \(\tilde{\mathcal{E}}^{(3)}\). The (green) edge \((v^{\dep}(A,2), v^{\arr}(B,3))\) is an edge of type \(\tilde{\mathcal{E}}^{(4)}\).

\paragraph{Constraint \eqref{eq: SelectingRoute}-\eqref{eq: PathConstraints}}
Here, we illustrate the constraints \eqref{eq: SelectingRoute}-\eqref{eq: PathConstraints} through an example. Consider an AAM vehicle \(u\) that wants to travel from region A to region C. Suppose the menu of that AAM vehicle is comprised of the routes \(s_1, s_2, s_3\), as described below:  
\begin{align*}
    s_1 &= \{(v(A,1),v(A,2)), (v(A,2),v(A,3)), (v(A,3),v(A,4)), (v(A,4),v(A,5))\}\quad  \texttt{Black-path in Figure \ref{fig:sector-price}} \\ 
    s_2 &= \{ (v(A,1), v^{\dep}(A,1)), (v^\dep(A,1), v^{\arr}(B,2)), (v^\arr(B,2), v(B,2)), (v(B,2), v(B,3)), (v(B,3), v(B,4)), \\ &\hspace{4cm} (v(B,4), v^\dep(B,4)), ( v^\dep(B,4), v^{\arr}(C,5)), \\ &\hspace{4cm} (v^{\arr}(C,5), v(C,5))\} \quad \texttt{{\color{red}red-path} in Figure \ref{fig:sector-price}} \\ 
    s_3 &= \{ (v(A,1),v(A,2)), (v(A,2), v^{\dep}(A,2)), (v^\dep(A,2), v^{\arr}(B,3)), (v^\arr(B,3), v(B,3)), \\ &\hspace{4cm} (v(B,3), v(B,4)), (v(B,4), v(B,5)), (v(B,5), v^\dep(B,5)), \\ &\hspace{4cm} (v^\dep(B,5), v^{\arr}(C,6))\} \quad \texttt{{\color{ForestGreen}green-path} in Figure \ref{fig:sector-price}}
\end{align*}
Here \(e^\ast(s_1)  = (v(A,1), v(A,2))\), \(e^\ast(s_2) = (v(A,1), v(A,1)^{\dep})\), \(e^\ast(s_3) = (v(A,2), v^{\dep}(A,2))\). The menu \(M_u\) of AAM vehicle \(u\) is given by \(
    M_u = \{s_1, s_2, s_3, \varnothing\}.\)

Consequently, the constraint \eqref{eq: SelectingRoute} for this AAM vehicle is given by \(
    x_{u, e^\ast(s_1)}+x_{u, e^\ast(s_2)}+x_{u, e^\ast(s_3)}  + x_{u,\varnothing} = 1.
\)
Additionally, the constraint \eqref{eq: PathConstraints} for this AAM vehicle contains two types of constraints: (1) the flow balance constraints: 
\begin{align*}
    & 
    x_{u,(v(A,1), v(A,2))} = x_{u, (v(A,2), v^\dep(A,2))} + x_{u, (v(A,2), v(A,3))} \\
&x_{u, (v(A,2), v(A,3))} = x_{u, (v(A,3), v(A,4))} \\
&x_{u, (v(A,3), v(A,4))} = x_{u, (v(A,4), v(A,5))}
\\
&x_{u, (v(A,2), v^\dep(A,2))} = x_{u,(v^\dep(A,2), v^\arr(B,3)} \\
& x_{u,(v^\dep(A,2), v^\arr(B,3)} = x_{u, (v^\arr(B,3), v(B,3))} \\ 
&  x_{u, (v^\arr(B,3), v(B,3))} + x_{u,(v(B,2), v(B,3))}  = x_{u, (v(B,3), v(B,4))}  \\ 
&x_{u, (v(B,3), v(B,4))} = x_{u,(v(B,4),v(B,5))} + x_{u,(v(B,4), v^\dep(B,4))} \\ 
& x_{u,(v(B,4),v(B,5))}  = x_{u,(v(B,2), v^\dep(B,5))} 
\\ 
& x_{u,(v(B,4), v^\dep(B,4))}  = x_{u,(v^\dep(B,4), v^\arr(C,5))} \\ 
& x_{u,(v^\dep(B,4), v^\arr(C,5))}  = x_{u,(v^\arr(C,5), v(C,5))} \\ 
&x_{u,(v^\dep(A,1), v^\arr(B,2))}  = x_{u,(v^\arr(B,2), v(B,2))} \\ 
&x_{u,(v^\arr(B,2), v(B,2))} = x_{u,(v(B,2), v(B,3))}.
\end{align*}
and (2) additional constraints:
\begin{align*}
    &x_{u,(v(A,1),v^\dep(A,1))} = x_{u,(v^\dep(B,4), v^\arr(C,5))} \\ 
    &x_{u,(v(A,2), v^\dep(A,2))} = x_{u,(v^\dep(B,5), v^\arr(C,6))} .
\end{align*}
 These constraints ensure that the path flows allocated on the departing edge, as per \eqref{eq: SelectingRoute}, result in unique edge flows on the entire network.

%% file: Appendix/ProofExistence.tex
Before presenting the proof, let us recall some important mathematical definitions and results which are crucial for the proof. First, we recall the definition of upper semicontinuous and lower semicontinuous correspondences. 
\begin{definition} A correspondence \(f:X\rightrightarrows Y\) is \emph{upper semicontinuous} if for every sequence \(x_n\in X\) (with limit \(x\)) and the sequence \(y_n\in f(x_n)\) which has a limit,  then there exists \(y\in f(x)\) such that \(y=\lim_n y_n\). 
\end{definition}
\begin{definition} A correspondence \(f:X\rightrightarrows Y\) is \emph{lower semicontinuous} if for every sequence \(x_n\in X\) (with limit \(x\)) and \(y\in f(x)\), then there exists a convergent sequence \(y_n \in f(x_n)\) with limit \(y\).
\end{definition}
Next, we recall the Kakutani fixed point theorem. 
\begin{theorem}[Kakutani Fixed Point Theorem]\label{thm: Kakutani}
    Suppose \(X\) is a non-empty, convex, and compact subset of  \(\mathbb{R}^n\) and \(f: X\rightrightarrows X\) is a non-empty, compact-valued, convex-valued, and upper semicontinuous correspondence. Then \(f\) has a fixed point. 
\end{theorem}
Finally, we recall Berge's maximum theorem. 
\begin{theorem}[Berge's Maximum Theorem]\label{thm: BergeMaximumTheorem}
    Consider the optimization problem \(
        \max_{x\in A(\theta)} F(x,\theta)\). 
    Let \(X(\theta)\) be the set of solutions of the preceding problem. If \(F\) is continuous in \((x,\theta)\) and \(\theta\rightrightarrows A(\theta)\) is a non-empty, compact-valued, and continuous correspondence, then \(X(\theta)\) is a non-empty, compact-valued, and upper semicontinous correspondence.  
\end{theorem}

\begin{proof}[Proof of Proposition \ref{prop: MarketClearingExists}]
The proof builds on a result about the existence of a competitive equilibrium in Fisher markets with auxiliary inequality constraints \cite{jalota2023fisher}. Particularly, our proof accounts for auxiliary \emph{equality} constraints (resulting due to \eqref{eq: SelectingRoute}-\eqref{eq: PathConstraints}). 

In this result, we consider a relaxation of \eqref{eq: UserOpt}, by converting the integrality constraint to the positivity constraint.
Consider the relaxed individual optimization problem of every agent stated below:
\begin{subequations}\label{eq: UserOptProof}
\begin{align}
\max_{\bar{\mathbf{x}}_u} &\quad f_u(\bar{\mathbf{x}}_u) \label{eq: Obj_Relaxed}\\ 
    \text{s.t.} & \quad \mathbf{p}^{\top}\mathbf{x}_u + p_{\outsideOpt}x_{u,\outsideOpt} = w_u 
    \label{eq: BudgetRelaxed}\\ 
{} & \quad \tilde{\mathbf{a}}_u^\top \mathbf{x}_u + x_{u,\varnothing}= 1
\label{eq: ConstraintsRelaxed_Ones}
\\ 
{} & \quad \tilde{\mathbf{A}}_u \mathbf{x}_u = \mathbf{0}
\label{eq: ConstraintsRelaxedTwo}
\\
    {}& \quad   x_{u,\outsideOpt}\geq 0, x_{u,\dropOut}\geq 0, x_{u,e} \geq 0 \quad \forall e\in\tilde{\mathcal{E}}\label{eq: IntegralRelaxed}.
\end{align}
\end{subequations}

To prove the existence result, we scale the problem such that the total budget of all agents is \(1\) and the capacity of each good is \(1\). 
To do this, for every \(e\in\tilde{\mathcal{E}}, u\in U\), we scale any allocation \(x_{u,e}\) to \(x_{u,e}/\ell_e\), scale \(\tilde{\mathbf{a}}_{u,e}\) to \(\tilde{\mathbf{a}}_{u,e}\cdot\ell_e\), scale \(\tilde{\mathbf{A}}_u[:,e]\) to \(\tilde{\mathbf{A}}_u[:,e]\cdot \ell_e\), \(p_e\) to \(p_e\ell_e/W\), and \(w_u\) to \(w_u/W\), where \(W = \sum_uw_u\). Note that under this change the solution of \eqref{eq: UserOptProof} does not change. 
Furthermore, due to the condition that \(v_{u,\outsideOpt} \geq 0\) and the variable \(x_{u,\outsideOpt}\) does not enter in the constraint \eqref{eq: ConstraintsRelaxed_Ones}-\eqref{eq: ConstraintsRelaxedTwo}, it is ensured that the budget constraint \eqref{eq: BudgetRelaxed} hold with equality. 

Define \(\Delta_{|\tilde{\mathcal{E}}|}= \{\mathbf{p}\in \mathbb{R}^{|\tilde{\mathcal{E}}|}: \sum_{e\in \tilde{\mathcal{E}}}p_e = 1, p_e\geq 0\ \forall e\in \tilde{\mathcal{E}}\}.\)
Moreover, for every UAV \(u\), define \(Y_u = \{\bar{\mathbf{x}}_u\in \mathbb{R}^{|\tilde{\mathcal{E}}|+2}_{\geq 0}: \tilde{\mathbf{a}}_u^\top \mathbf{x}_u + x_{u,\dropOut} =  1, 
   \quad \bar{\mathbf{A}}_u\mathbf{x}_u =  0\},\) and \(Q_u = \{\bar{\mathbf{x}}_u\in \mathbb{R}^{|\tilde{\mathcal{E}}|+2}_{\geq 0}: x_{u,r}\leq \Omega, \quad \forall \ r \in \tilde{\mathcal{E}}\cup\{\outsideOpt, \dropOut\}\}\) for some \(\Omega> 1\). Define \(X = \Pi_{u\in U} Q_u\).  

Define \(B_u(\mathbf{p}) = \{\bar{\mathbf x}_u\in Y_i: \mathbf{p}^\top \mathbf{x}_u + p_{\outsideOpt}x_{u,\outsideOpt} = w_u\}\). Note that this set is non-empty, so we can always choose \(x_{u,\dropOut}\) to ensure that \(\mathbf{x}_u = \mathbf{0}\) and spend all the budget in the outside option \(\outsideOpt\).  
Define 
\begin{align}
\tilde{\mathbf{x}}_u(\mathbf{p}) &= \underset{\bar{\mathbf{x}}_u\in Q_u\cap B_u(\mathbf{p})}{\arg\max}f_u(\bar{\mathbf{x}}_u), \label{eq: x_Mapping}
\\
\tilde{\mathbf{p}}(\mathbf{x}) &= \underset{\mathbf{p}\in \Delta_{|\tilde{\mathcal{E}}|}}{\arg\max}~\mathbf{p}^\top\left(\sum_{u\in U}\mathbf{x}_u - \mathbf{1}\right)\label{eq: P_mapping}. 
\end{align}
Using the above definitions, define a correspondence \(h(\mathbf{x},\mathbf{p}) = ((\tilde{\mathbf{x}}_u(\mathbf{p}))_{u\in U}, \tilde{\mathbf{p}}(\mathbf{x}))\). We shall show that a fixed point of this mapping exists and is a fractional competitive equilibrium. 

\paragraph{Existence of a Fixed Point}

We show that \(h\) satisfies the condition of the Kakutani fixed point theorem (cf. Theorem \ref{thm: Kakutani}), which ensures the existence of a fixed point. First, note that the domain of \(h\), i.e. \(X\times \Delta_{|\tilde{\mathcal{E}}|}\), is non-empty, compact and convex. 

Next, we show that \(h\) is a non-empty, compact-valued, convex-valued, and upper semicontinuous correspondence. It is enough to show that \(\tilde{\mathbf{x}}_u(\mathbf{p})\) and \(\tilde{\mathbf{p}}(\mathbf{x})\) are non-empty, convex-valued and upper semicontinuous correspondences. 

From \eqref{eq: P_mapping}, we observe that \(\tilde{\mathbf{p}}(\mathbf{x})\) is non-empty and convex-valued and is an optimal solution to a linear program with a non-empty, convex, and compact feasible set. All conditions for Berge's maximum theorem (cf. Theorem \ref{thm: BergeMaximumTheorem}) are satisfied, and therefore \(\tilde{\mathbf{p}}(\mathbf{x})\) is also compact-valued and upper semicontinuous.

Next, we show that \(\tilde{\mathbf{x}}_u(\mathbf{p})\) is non-empty and convex-valued as it is the optimizer of a linear function on a non-empty, convex, and compact set. Next, we leverage Theorem \ref{thm: BergeMaximumTheorem} to show that this map is compact-valued and upper semicontinuous. First, we need to show that the correspondence \(g_u: \mathbf{p} \rightrightarrows Q_u\cap B_u(\mathbf{p})\) is a compact-valued and continuous correspondence. Compactness follows by construction, so the only thing remaining to show is continuity. To show continuity, it is enough to show that the mapping is upper semicontinuous and lower semicontinuous. 

To show \(g_u\) is upper semicontinuous, consider a sequence \((\bar{\mathbf{x}}_u^n,\mathbf{p}^n)\) such that \(\bar{\mathbf{x}}_u^n \in Q_u\cap B_u(\mathbf{p}^n)\), which has limit \((\bar{\mathbf{x}}_u, \mathbf{p})\). Then, it is sufficient to establish that \(\bar{\mathbf{x}}_u\in Q_u\cap B_u(\mathbf{p})\). Note that \(Q_u\) is compact, so if \(\bar{\mathbf{x}}_u^n\in Q_u\), for every \(n\in \mathbb{N}\), it follows that \(\bar{\mathbf{x}}_u\in Q_u\). Furthermore, since \(\bar{\mathbf{x}}_u^n\geq 0\), for every \(n\in \mathbb{N}\),  it follows that \(\bar{\mathbf{x}}_u\geq 0\). Additionally, for every \(n\in \mathbb{N}\), \(\tilde{\mathbf{a}}_u^\top{\mathbf{x}}_u^n + x_{u,\varnothing}^n = 1, 
   \quad \tilde{\mathbf{A}}_u{\mathbf{x}}_u^n =  \mathbf{0}\), it follows that \(\tilde{\mathbf{a}}_u^\top{\mathbf{x}}_u + x_{u,\varnothing} = 1, 
   \quad \bar{\mathbf{A}}_u{\mathbf{x}}_u =  \mathbf{0}\). Moreover, the continuity of product ensures that \(\mathbf{p}^n{}^\top \mathbf{x}_u^n + p_\outsideOpt x_{u,\outsideOpt}^n = w_u\), for every \(n\in \mathbb{N}\), implies \(\mathbf{p}{}^\top \mathbf{x}_u + p_\outsideOpt x_{u,\outsideOpt} = w_u\). This ensures that \(g_u\) is upper semicontinuous.

Next, we show that \(g_u\) is lower semicontinuous. To show this, it is sufficient to show that for any sequence \(\mathbf{p}^n\) with limit \(\mathbf{p}\) and any point \(\bar{\mathbf{x}}_u\in Q_u\times B_u(\mathbf{p})\)
 there is a sequence \(\bar{\mathbf{x}}_u^n\in Q_u\cap B_u(\mathbf{p}^n)\) such that \(\lim_{n\rightarrow \infty}\bar{\mathbf{x}}_u^n = \bar{\mathbf{x}}_u\).
Towards this goal, for every \(u\in U, e\in \tilde{\mathcal{E}}\), we define \(\bar{\mathbf{x}}^n\) such that 
 \begin{align*}
    {x}_{u,e}^n &= \min\left\{1,\frac{w_u}{\mathbf{p}^n{}^\top \mathbf{x}_u+p_{\outsideOpt}x_{u,\outsideOpt}}\right\}x_{u,e},\quad   x_{u,\dropOut}^n = 1- \tilde{\mathbf{a}}_u^\top{\mathbf{x}}_u^n, \quad  x_{u,\outsideOpt}^n  =  \frac{1}{p_\outsideOpt}\left(w_u - \mathbf{p}^n{}^\top \mathbf{x}_u^n\right). 
 \end{align*}
It is easy to check that \(\lim_{n\rightarrow \infty}\bar{\mathbf{x}}_u^n = \bar{\mathbf{x}}_u\). The only thing remaining to show is that 
\(\bar{\mathbf{x}}_u^n\in Q_u\cap B_u(\mathbf{p}^n)\). First, note that \(\tilde{\mathbf{a}}_u^\top\mathbf{x}_u^n + x_{u,\dropOut}^n =  1\) follows by construction. Next, we show that \(x_{u,\dropOut}^n\geq 0
\). This is because 
\begin{align*}
    &\tilde{\mathbf{a}}_u^\top\mathbf{x}_u^n = \min\left\{1,\frac{w_u}{\mathbf{p}^n{}^\top \mathbf{x}_u+p_{\outsideOpt}x_{u,\outsideOpt}}\right\}\tilde{\mathbf{a}}_u^\top\mathbf{x}_u \leq \tilde{\mathbf{a}}_u^\top\mathbf{x}_u = 1 -x_{u,\dropOut} \\ 
    &\implies 
   x_{u,\dropOut}^n = 1 - \tilde{\mathbf{a}}_u^\top\mathbf{x}_u^n  \geq x_{u,\dropOut} \geq 0, 
\end{align*}
where the inequality follows as \(\tilde{a}_{u,e}\geq 0, x_{u,e}^n\geq 0\). 
Similarly, one can show that \(\tilde{\mathbf{A}}_u\mathbf{x}_u^n = \mathbf{0}\). 
Next, we note that budget constraints are satisfied by the construction of \(x_{u,\outsideOpt}^{n}\). Finally, we show that \(x_{u,\outsideOpt}^n\geq 0\). Indeed, 
\begin{align*}
    \mathbf{p}^n{}^\top \mathbf{x}_u^n =\min\left\{1,\frac{w_u}{\mathbf{p}^n{}^\top \mathbf{x}_u+p_{\outsideOpt}x_{u,\outsideOpt}}\right\}p^n{}^\top \mathbf{x}_u  \leq w_u.  
\end{align*}
Thus, we conclude that \(g_u\) is a compact-valued continuous correspondence. 
Thus, from Theorem \ref{thm: Kakutani},  we conclude that there exists  \((\bar{\mathbf{x}}^\ast,\mathbf{p}^\ast)\) such that \(
\bar{\mathbf{x}}^\ast_u = \tilde{\mathbf{x}}_u(\mathbf{p}^\ast), \quad \mathbf{p}^\ast = \tilde{\mathbf{p}}(\bar{\mathbf{x}}^\ast), \quad \forall \ u\in U.
\)

\paragraph{Existence of a Fractional Competitive Equilibrium}

We show that any fixed point corresponds to a fractional competitive equilibrium. 
First, using \eqref{eq: x_Mapping}, we conclude that \(\bar{\mathbf{x}}^\ast_u\) is an optimal solution to \eqref{eq: UserOptProof}. Second, note that \(\mathbf{p}^\ast \in \mathbb{R}_{\geq 0}^{|\tilde{\mathcal{E}}|}\) by construction.  Next, we show that the capacity constraints are satisfied. We show this by contradiction. Suppose there exists an edge \(e'\in\tilde{\mathcal{E}}\) such that \(\sum_{u\in U}x_{u,e
'}^\ast > 1\). Then by \eqref{eq: P_mapping}, it must hold that 
\begin{align*}
    \sum_{e\in \tilde{\mathcal{E}}}p_e^\ast\left(\sum_{u\in U}x_{u,e}^\ast - 1\right)  \geq \sum_{e\in \tilde{\mathcal{E}}}p_e\left(\sum_{u\in U}x_{u,e}^\ast - 1\right), \quad \forall \ \mathbf{p}\in \Delta_{|\tilde{\mathcal{E}}|}. 
\end{align*}
We claim that \(\sum_{e\in \tilde{\mathcal{E}}}p_e^\ast (\sum_{u\in U}x_{u,e}^\ast - 1)= 0\). Indeed, 
\begin{align*}
\sum_{e\in \tilde{\mathcal{E}}}p_e^\ast (\sum_{u\in U}x_{u,e}^\ast - 1)= \sum_{u\in U}\sum_{e\in\tilde{\mathcal{E}}} p_e^\ast x_{e,u}^\ast  - \sum_{e\in \tilde{\mathcal{E}}} p_e^\ast = \sum_{u\in U}w_u - 1 = 0.   
\end{align*}
Thus, we conclude that 
\begin{align}\label{eq: p_eq}
   0 \geq  \sum_{e\in \tilde{\mathcal{E}}}p_e\left(\sum_{u\in U}x_{u,e}^\ast - 1\right), \quad \forall \ \mathbf{p}\in \Delta_{|\tilde{\mathcal{E}}|}. 
\end{align}
Since  \(\sum_{u\in U}x_{u,e
'}^\ast > 1\), we can select \(p_{e'} = 1\) and \(0\) otherwise, which would violate the above inequality, a contradiction. 

Next, we show that if \(p_e^\ast > 0\) then \(\sum_{u\in U}x_{u,e}^\ast = 1\). This follows immediately from the fact that capacity constraints are satisfied and the fact that \(\sum_{e\in \tilde{\mathcal{E}}} p_e^\ast (\sum_{u\in U}x_{u,e}^\ast - 1) = 0\). This completes the proof. 
\end{proof}

%% file: Appendix/ProofComputation.tex
Observe that for any fixed value of \(\omega\in \mathbb{R}^{|U|}\), the optimization problem \eqref{eq: PlannerOptMain} is a convex optimization problem. Define the Lagrangian as follows. 
\begin{align*}
\mathcal{L}_{\mathbf{P}} &= \sum_{u\in U}(w_u + \omega_u) \log\left(f_u(\bar{\mathbf{x}}_u)\right)-\sum_{u\in U} p_{\outsideOpt}x_{u,\outsideOpt} - \mathbf{p}^\top\left(\sum_{u\in U}\mathbf{x}_u-  \ell \right) \\ & - \sum_{u\in U}\lambda_u(\tilde{\mathbf{a}}_u^\top \mathbf{x}_u  + x_{u,\varnothing} - 1 
) - \sum_{u\in U}\kappa_u^\top \tilde{\mathbf{A}}_{u}\mathbf{x}_u +\sum_{u\in U}\mu_{u}^\top \bar{\mathbf{x}}_u, 
\end{align*}
where \(\mathbf{p}\in \mathbb{R}^{|\tilde{\mathcal{E}}|}_{\geq 0}\) is the Lagrange multiplier corresponding to constraint \eqref{eq: Cap_ConstraintM}, \(\lambda=(\lambda_u)_{u\in U} \in \mathbb{R}^{|U|}\) is the Lagrange multiplier corresponding to \eqref{eq: SelectingRouteM}, \(\kappa = (\kappa_u)_{u\in U} \in \mathbb{R}^{K|U|}\) is the Lagrange multiplier corresponding to \eqref{eq: PathConstraintsM}, and \(\mu = (\mu_{u})_{u \in U} \in \mathbb{R}^{|U||\tilde{\mathcal{E}}|}_{\geq 0}\) is the Lagrange multiplier corresponding to \eqref{eq: Pos_constraintM}.  
 
We observe that, for a given \(\omega\), any optimal solution \(\bar{\mathbf{x}}^\dagger\) of \eqref{eq: PlannerOptMain} with optimal dual multipliers \((\mathbf{p}^\dagger, \lambda^\dagger, \kappa^\dagger, \mu^\dagger)\) will satisfy the following first order conditions of optimality. 
{\begin{equation}\label{eq: FOS_BASOP}
\begin{aligned}
    0 \geq  \begin{cases}
        \frac{(w_u+\omega_u)}{f_u(\bar{\mathbf{x}}_u^\dagger)} v_{u,e} - p_e^\dagger - \tilde{a}_{u,e} \lambda_u^\dagger - (\tilde{\mathbf{A}}_u^\top\kappa_u^\dagger)_e& \text{if} \ e\in \tilde{\mathcal{E}}\\ 
        \frac{(w_u+\omega_u)}{f_u(\bar{\mathbf{x}}_u^\dagger)} v_{u,\outsideOpt} - p_{\outsideOpt}& \text{if} \ e = \outsideOpt \\ 
         \frac{(w_u+\omega_u)}{f_u(\bar{\mathbf{x}}_u^\dagger)} v_{u,\dropOut} - \lambda_u^\dagger   & \text{if} \ e = \dropOut.  
     \end{cases}
\end{aligned}
\end{equation}}
Furthermore, the complementary slackness conditions are given by 
\begin{equation}\label{eq: CompSlack_BASOP}
\begin{aligned}
    0 =  \begin{cases}
        \frac{(w_u+\omega_u)}{f_u(\bar{\mathbf{x}}_u^\dagger)} v_{u,e}x_{u,e}^\dagger - p_e^\dagger x_{u,e}^\dagger - \tilde{a}_{u,e}x_{u,e}^\dagger \lambda_u^\dagger - (\tilde{\mathbf{A}}_u^\top\kappa_u^\dagger)_ex_{u,e}^\dagger& \text{if} \ e\in \tilde{\mathcal{E}}\\ 
        \frac{(w_u+\omega_u)}{f_u(\bar{\mathbf{x}}_u^\dagger)} v_{u,\outsideOpt}x_{u,\outsideOpt}^\dagger - p_{\outsideOpt}^\dagger x_{u,\outsideOpt}^\dagger& \text{if} \ e = \outsideOpt \\ 
         \frac{(w_u+\omega_u)}{f_u(\bar{\mathbf{x}}_u^\dagger)} v_{u,\dropOut}x_{u,\dropOut}^\dagger - \lambda_ux_{u,\dropOut}^\dagger   & \text{if} \ e = \dropOut  
     \end{cases}, \quad p_e^\dagger(\sum_{u\in U}x_{u,e}^\dagger - \ell_e) = 0, \quad \forall \ e\in \tilde{\mathcal{E}}. 
\end{aligned}
\end{equation}

Similarly, the Lagrangian of the (relaxed) individual optimization problem \eqref{eq: UserOptProof} is given by 
\begin{align*}
    \mathcal{L}_{\mathbf{I}} &= f_u(\bar{\mathbf{x}}_u) - \tilde{\omega}_u\left(\mathbf{p}^\top \mathbf{x}_u + p_{\outsideOpt}x_{u,\outsideOpt} -w_u\right) - \tilde{\lambda}_u(\tilde{\mathbf{a}}_u^\top \mathbf{x}_u  + x_{u,\varnothing} - 1 
) - \sum_{u\in U}\tilde{\kappa}_u^\top \tilde{\mathbf{A}}_{u}\mathbf{x}_u + \tilde{\mu}_u^\top \bar{\mathbf{x}}_u, 
\end{align*} 
where \(\tilde{\omega}_u \in \mathbb{R}\) is the Lagrange multiplier corresponding to the budget constraint, \(\tilde{\lambda}_u \in \mathbb{R}\) is the Lagrange multiplier corresponding to \eqref{eq: ConstraintsRelaxed_Ones}, \(\tilde{\kappa}_u \in \mathbb{R}^{K}\) is the Lagrange multiplier corresponding to \eqref{eq: ConstraintsRelaxedTwo}, and \(\tilde{\mu}_u \in \mathbb{R}^{|\tilde{\mathcal{E}}|}_{\geq 0}\) is the Lagrange multiplier corresponding to the positivity constraint \eqref{eq: IntegralRelaxed}. 

We observe that, for a given \(\mathbf{p}\), any optimal solution \(\bar{\mathbf{x}}^\ddagger\) of \eqref{eq: UserOptProof} with optimal dual multipliers \((\tilde\omega^\ddagger, \tilde\lambda^\ddagger, \tilde\kappa^\ddagger, \tilde\mu^\ddagger)\) satisfies the following first order conditions of optimality. 
\begin{equation}\label{eq: FOS_IOP}
\begin{aligned}
     0 \geq  
     \begin{cases}
      v_{u,e}-\tilde{\omega}_u^\ddagger p_e - \tilde{\lambda}_u^\ddagger \tilde{a}_{u,e} - (\tilde{\mathbf{A}}_u^\top \tilde{\kappa}_u^\ddagger)_e  & \text{if} \ e\in \tilde{\mathcal{E}}\\ 
        v_{u,\outsideOpt} - \tilde{\omega}_u^\ddagger p_{\outsideOpt}& \text{if} \ e = \outsideOpt \\ 
         v_{u,\dropOut} - \tilde{\lambda}_u^\ddagger & \text{if} \ e = \dropOut. 
     \end{cases}
\end{aligned}
\end{equation}
Furthermore, using the complementary slackness condition, we obtain 
\begin{equation}\label{eq: CompSlack_IOP}
\begin{aligned}
     0 = 
     \begin{cases}
       v_{u,e}x_{u,e}^\ddagger-\tilde{\omega}_u^\ddagger p_ex_{u,e}^\ddagger  - \tilde{\lambda}_u^\ddagger \tilde{a}_{u,e}x_{u,e}^\ddagger - (\tilde{\mathbf{A}}_u^\top \tilde{\kappa}_u^\ddagger)_ex_{u,e}^\ddagger & \text{if} \ e\in \tilde{\mathcal{E}}\\ 
        v_{u,\outsideOpt}x_{u,\outsideOpt}^\ddagger - \tilde{\omega}_u^\ddagger p_{\outsideOpt}x_{u,\outsideOpt}^\ddagger& \text{if} \ e = \outsideOpt \\ 
         v_{u,\dropOut}x_{u,\dropOut}^\ddagger - \tilde{\lambda}_u^\ddagger x_{u,\dropOut}^\ddagger & \text{if} \ e = \dropOut. 
     \end{cases}
\end{aligned}
\end{equation}


In order to prove Proposition \ref{prop: FixedPointOptimization}, we show
that if there exists \(\omega^\ast\) such that \(\omega^\ast = \lambda^\dagger(\omega^\ast)\) then \((\bar{\mathbf{x}}^\dagger(\omega^\ast), \mathbf{p}^\dagger(\omega^\ast))\) is a fractional-competitive equilibrium. It is sufficient to verify the following: 
\begin{itemize}
    \item[(i)] By fixing the prices to  \(\mathbf{p}^\dagger(\omega^\ast),\) \(\bar{\mathbf{x}}_u^\dagger(\omega^\ast)\) is an optimal solution of \eqref{eq: UserOptProof}, for every \(u\in U\);
    \item[(ii)] the capacity constraints are satisfied at every resource; 
    \item[(iii)]  \(p_e^\dagger(\omega^\ast)\geq 0\)  for every \(e\in \tilde{\mathcal{E}}\); and
    \item[(iv)] if \(p_e^\dagger(\omega^\ast) > 0\) for some \(e\in \tilde{\mathcal{E}}\), then \(\sum_{u\in U} {x}_{u,e}^\dagger(\omega^\ast) = \ell_e\).
\end{itemize}
It is immediate to note that \((ii)-(iv)\) are satisfied due to dual and primal feasibility conditions of \eqref{eq: PlannerOptMain}. It only remains to show \((i)\). 

To show \((i)\), it is sufficient to show that the there exists \(\tilde{\omega}_u^\ddagger, \tilde{\lambda}_u^\ddagger, \tilde{\kappa}_u^\ddagger\) such that \((\bar{\mathbf{x}}_u^\dagger(\omega^\ast),\tilde{\omega}_u^\ddagger, \tilde{\lambda}_u^\ddagger, \tilde{\kappa}_u^\ddagger)\) satisfies the conditions \eqref{eq: FOS_IOP}-\eqref{eq: CompSlack_IOP}, and the budget constraint in \eqref{eq: BudgetRelaxed} holds. 

Setting the optimal Lagrange variable of \eqref{eq: Cap_ConstraintM} with \(\omega_u = \lambda_u\) then the optimal solution \(x^\ast\) of \eqref{eq: PlannerOptMain} is the solution of individual optimization problem for all players with price \(\mathbf{p}^\ast\). 

By primal optimality conditions in \eqref{eq: FOS_BASOP}, we obtain 
\begin{equation}
\begin{aligned}
    0 \geq  \begin{cases}
         v_{u,e} - \frac{f_u(\bar{\mathbf{x}}_u^\dagger)}{(w_u+\omega_u^\ast)}p_e^\dagger - \frac{f_u(\bar{\mathbf{x}}_u^\dagger)}{(w_u+\omega_u^\ast)}\tilde{a}_{u,e} \lambda_u^\dagger - \frac{f_u(\bar{\mathbf{x}}_u^\dagger)}{(w_u+\omega_u^\ast)}(\tilde{\mathbf{A}}_u^\top\kappa_u^\dagger)_e& \text{if} \ e\in \tilde{\mathcal{E}}\\ 
      v_{u,\outsideOpt} - \frac{f_u(\bar{\mathbf{x}}_u^\dagger)}{(w_u+\omega_u^\ast)} p_{\outsideOpt}& \text{if} \ e = \outsideOpt \\ 
          v_{u,\dropOut} - \frac{f_u(\bar{\mathbf{x}}_u^\dagger)}{(w_u+\omega_u^\ast)}\lambda_u^\dagger   & \text{if} \ e = \dropOut.  
     \end{cases}
\end{aligned}
\end{equation}
The preceding equation is equivalent to the primal optimality condition of individual optimization problem in \eqref{eq: FOS_IOP} if we select \(\tilde{\lambda}_u^\ddagger = \frac{f_u(\bar{\mathbf{x}}_u^\dagger)}{(w_u+\omega_u^\ast)} \lambda_u^\dagger\), \(\tilde{\omega}_u = \frac{f_u(\bar{\mathbf{x}}_u^\dagger)}{(w_u+\omega_u^\ast)}\) and \(\tilde{\kappa}_u^\ddagger =\frac{f_u(\bar{\mathbf{x}}_u^\dagger)}{(w_u+\omega_u^\ast)}\kappa_u^\dagger\). Similarly, \eqref{eq: CompSlack_IOP} is also satisfied with the same choice. 
Finally, we show that individual budget constraint \eqref{eq: BudgetRelaxed} holds. For this we use the complementary slackness condition in \eqref{eq: CompSlack_BASOP} by summing all three cases in \eqref{eq: CompSlack_BASOP}. For every \(u\in U,\) we obtain  
\begin{align*}
    0=\frac{(w_u+\omega_u^\ast)}{f_u(\bar{\mathbf{x}}_u^\dagger)} f_u(\bar{\mathbf{x}}_u^\dagger) - \mathbf{p}^\dagger{}^\top\mathbf{x}_u^\dagger - p_{\outsideOpt}\mathbf{x}_{u,\outsideOpt}^\dagger- \lambda_u^\dagger(\tilde{\mathbf{a}}_u^\top\mathbf{x}_u^\dagger + x_{u,\dropOut}^\dagger) - \kappa_u^\dagger{}^\top\tilde{\mathbf{A}}_u^\top \mathbf{x}_u^\dagger.
\end{align*}
Consequently, using \eqref{eq: SelectingRouteM}-\eqref{eq: PathConstraintsM} we obtain 
\begin{align*}
    0 &= (w_u+\omega_u^\ast) - \mathbf{p}^\dagger{}^\top\mathbf{x}_u^\dagger - p_{\outsideOpt}x_{u,\outsideOpt}^\dagger - \lambda_u^\dagger \\ &= w_u- \mathbf{p}^\dagger{}^\top\mathbf{x}_u^\dagger - p_{\outsideOpt}x_{u,\outsideOpt}^\dagger, 
\end{align*}
where in the last equation we used the fact that \(\omega^\ast= \lambda^\dagger\). 
This completes the proof.

%% file: Appendix/ADMM_generalMethod.tex
 
The updates in the inner loop in Algorithm \ref{algo:1} is derived based on ADMM updates for \eqref{eq: PlannerOptADMM}. 
We review the basic structure of the ADMM algorithm in Section \ref{ssec: ADMM_General} and then derive the inner loop updates in Section \ref{ssec: Derivation}.

\subsection{Review of ADMM algorithm}\label{ssec: ADMM_General}
The Alternative Direct Method of Multipliers (ADMM) is a distributed convex optimization algorithm that decomposes a problem into smaller subproblems, solves them in parallel, and coordinates to find a global solution via dual updates \cite{BOYLES2010519, Boyd_Vandenberghe_2004}. It is built on dual ascent and augmented lagrangian methods. 

Consider the following optimization problem with separable cost structure: 
\begin{equation}\label{eq: ADMM_General}
\begin{aligned}
    \max_{\mathbf{x}\in X, \mathbf{y}\in Y} & \ h(\mathbf{x}, \mathbf{y}) = h_1(\mathbf{x}) + h_2(\mathbf{y}) \\ 
    \text{s.t.} & \ A\mathbf{x} + B\mathbf{y} = \mathbf{c},
\end{aligned}
\end{equation}
where
\begin{itemize}
    \item[(i)]  \(X\subset \mathbb{R}^a, Y\subset \mathbb{R}^b\) are closed convex sets, 
    \item[(ii)] \(h_1: \mathbb{R}^a \rightarrow \mathbb{R}, h_2: \mathbb{R}^b \rightarrow \mathbb{R},\)
    \item[(iii)] \(A\in \mathbb{R}^{s \times a}, B\in \mathbb{R}^{s \times b}, \mathbf{c}\in \mathbb{R}^s.\) 
\end{itemize}

Let \(\mu\in \mathbb{R}^s\) be the dual multiplier of constraint in \eqref{eq: ADMM_General}. Consider the following augmented Lagrangian function for \eqref{eq: ADMM_General} for some parameter \(\beta> 0\)
\begin{align*}
    L_{\beta}(\mathbf{x}, \mathbf{y}) = h_1(\mathbf{x}) + h_2(\mathbf{y}) - \mu^\top (A\mathbf{x} + B\mathbf{y} - \mathbf{c}) - \frac{\beta}{2}\|A\mathbf{x} + B\mathbf{y} - \mathbf{c}\|^2.
\end{align*}

The ADMM algorithm is a discrete-time algorithm, indexed by \(k\), given as follows
\begin{equation}\label{eq: ADMM_Updates_General}
    \begin{aligned}
        \mathbf{x}^{(n+1)} &= \arg\max_{\mathbf{x}\in X} L_{\beta}(\mathbf{x}, \mathbf{y}^{(n)}) \\ 
        \mathbf{y}^{(n+1)} &= \arg\max_{\mathbf{y}\in Y} L_{\beta}(\mathbf{x}^{(n+1)}, \mathbf{y}) \\ 
        \mu^{(n+1)} &= \mu^{(n)} + \beta(A\mathbf{x}^{(n+1)} + B\mathbf{y}^{(n+1)}-\mathbf{c}).
    \end{aligned}
\end{equation}
The parameter \(\beta\) is also referred to as the step-size parameter for the ADMM algorithm.

\subsection{ADMM Updates for \eqref{eq: PlannerOptADMM}}\label{ssec: Derivation}
The inner loop in Algorithm \ref{algo:1} is nothing but the 
ADMM algorithm applied to \eqref{eq: PlannerOptADMM}. 

For any \(\beta>0\), we form the augmented Lagrangian $\mathcal{L}(\bar{\mathbf{x}}, \mathbf{y}, \mathbf{z}, \lambda, \mathbf{p}, \tilde{\mathbf{p}})$ for \eqref{eq: PlannerOptADMM} as follows 

\begin{align}
L_{\beta}(\bar{\mathbf{x}}, \mathbf{y}, \mathbf{z}, \lambda, \mathbf{p}, \tilde{\mathbf{p}}) = 
& \sum_{u \in U} (w_u + \omega_u) \log (f_u (\bar{\mathbf{x}}_u)) 
- \sum_{u\in U} p_{\outsideOpt}x_{u,\outsideOpt} 
- \sum_{u\in U}\tilde{\mathbf{p}}_u^\top(\mathbf{x}_u-\mathbf{y}_u)
- \mathbf{p}^\top \left( \sum_{u \in U} \mathbf{y}_u + \mathbf{z} - \ell \right) \nonumber \\ &  - \sum_{u \in U} \lambda_u (\tilde{\mathbf{a}}_u^\top \mathbf{x}_u + x_{u, \varnothing} - 1) - \frac{\beta}{2} \sum_{u \in U} \|\mathbf{x}_u - \mathbf{y}_u\|^2 - \frac{\beta}{2} \left\| \sum_{u \in U} \mathbf{y}_u + \mathbf{z} - \ell \right\|^2.
\end{align}

The ADMM algorithm (as per \eqref{eq: ADMM_Updates_General}) are given as follows:
\begin{subequations}\label{eq: ADMM_General_Updates}
 \begin{align}
    \bar{\mathbf{x}}^{(n+1)} &= \argmax_{\bar{\mathbf{x}}, \ \text{s.t.}~\eqref{eq: PathConstraintsM}-\eqref{eq: Pos_constraintM}~\text{hold}} L_{\beta}(\bar{\mathbf{x}}, \mathbf{y}^{(n)}, \mathbf{z}^{(n)}, \lambda^{(n)}, \mathbf{p}^{(n)}, \tilde{\mathbf{p}}^{(n)})\notag \\
 &= \argmax_{\bar{\mathbf{x}}, \ \text{s.t.}~\eqref{eq: PathConstraintsM}-\eqref{eq: Pos_constraintM}~\text{hold}} \quad \sum_{u \in U} (w_u + \omega_u) \log (f_u (\bar{\mathbf{x}}_u)) 
- \sum_{u\in U} p_{\outsideOpt}x_{u,\outsideOpt} 
- \sum_{u\in U}\tilde{\mathbf{p}}_u^{(n)}{}^\top(\mathbf{x}_u-\mathbf{y}_u^{(n)}) \notag   \\
&\hspace{3cm}
- \sum_{u \in U} \lambda_u^{(n)} (\tilde{\mathbf{a}}_u^\top \mathbf{x}_u + x_{u, \varnothing} - 1) - \frac{\beta}{2} \sum_{u \in U} \|\mathbf{x}_u - \mathbf{y}_u^{(n)}\|^2 \label{subeq: ADMM_General_x_update}
\\
({\mathbf{y}}^{(n+1)},\mathbf{z}^{(n+1)})&=\argmax_{{\mathbf{y}}\in \mathbb{R}^{U \times |\tilde{\mathcal{E}}|}, \ {\mathbf{z}} \in \mathbb{R}^{|\tilde{\mathcal{E}}|}_+}  L_{\beta}(\bar{\mathbf{x}}^{(n+1)}, \mathbf{y}, \mathbf{z}, \lambda^{(n)}, \mathbf{p}^{(n)}, \tilde{\mathbf{p}}^{(n)}) 
\notag \\
&= \argmax_{{\mathbf{y}}\in \mathbb{R}^{U \times |\tilde{\mathcal{E}}|}, \ {\mathbf{z}} \in \mathbb{R}^{|\tilde{\mathcal{E}}|}_+} -\sum_{u\in U}\tilde{\mathbf{p}}_u^{(n)}{}^\top(\mathbf{x}_u^{(n+1)}-\mathbf{y}_u)
- \mathbf{p}^{(n)}{}^\top \left( \sum_{u \in U} \mathbf{y}_u + \mathbf{z} - \ell \right) \notag \\ &\hspace{3cm} - \frac{\beta}{2} \sum_{u \in U} \|\mathbf{x}_u^{(n+1)} - \mathbf{y}_u\|^2 - \frac{\beta}{2} \left\| \sum_{u \in U} \mathbf{y}_u + \mathbf{z} - \ell \right\|^2 \label{subeq: ADMM_General_y_update}
\\ 
    \lambda^{(n+1)}_u &= \lambda_u + \beta(\tilde{\mathbf{a}}_u^\top\mathbf{x}_u^{(n+1)} + x_{u,\varnothing}^{(n+1)} - 1), \quad \forall \ u\in U \label{subeq: ADMM_General_lam_update}\\ 
    \mathbf{p}^{(n+1)} &= \mathbf{p}^{(n)}_u + \beta(\sum_{u\in U}\mathbf{y}_u^{(n+1)} + \mathbf{z}^{(n+1)}-\ell) \label{subeq: ADMM_General_p_update}\\ 
    \tilde{\mathbf{p}}^{(n+1)}_u &= \tilde{\mathbf{p}}^{(n)}_u + \beta(\mathbf{x}_u^{(n+1)} - \mathbf{y}_u^{(n+1)}), \quad \forall \ u\in U. \label{subeq: ADMM_General_ptil_update}
    \end{align}
    \end{subequations}

First, we claim the if \(\mathbf{p}^{(0)} = \tilde{\mathbf{p}}_u^{(0)}\) for every \(u\in U\), then \(\mathbf{p}^{(n)} = \tilde{p}_u^{(n)}\) for every \(u\in U\) and \(n\in \mathbb{N}\). We prove this by induction. Suppose for some \(n,\) \(\mathbf{p}^{(n)} = \tilde{p}_u^{(n)}\) for every \(u\in U\) then we show that \(\mathbf{p}^{(n+1)} = \tilde{p}_u^{(n+1)}\) for every \(u\in U.\)
To see this, note from the first order conditions of optimality for \eqref{subeq: ADMM_General_y_update} with respect to \(\mathbf{y}\), we obtain
\begin{align}\label{eq: EqualityP_PTil}
    \tilde{\mathbf{p}}_u^{(n)} - \mathbf{p}_u^{(n)} + \beta (\mathbf{x}_u^{(n+1)} - \mathbf{y}^{(n+1)}) - \beta (\sum_{u\in U}\mathbf{y}_u^{(n+1)} + \mathbf{z}^{(n+1)} - \ell) = 0. 
\end{align}
Using preceding equation, we obtain 
\begin{align*}
    \tilde{\mathbf{p}}_u^{(n+1)} \underset{\eqref{subeq: ADMM_General_ptil_update}}{=} \tilde{\mathbf{p}}_u^{(n)} + \beta (\mathbf{x}_u^{(n+1)} - \mathbf{y}^{(n+1)}) \underset{\eqref{eq: EqualityP_PTil}}{=} \mathbf{p}_u^{(n)} + \beta (\sum_{u\in U}\mathbf{y}_u^{(n+1)} + \mathbf{z}^{(n+1)} - \ell) \underset{\eqref{subeq: ADMM_General_p_update}}{=} {\mathbf{p}}_u^{(n+1)}. 
\end{align*}
This concludes our claim. Therefore, we get rid of notation \(\tilde{\mathbf{p}}\) and work only with \(\mathbf{p}.\)

Finally, note that \eqref{subeq: ADMM_General_x_update} is separable in \(\bar{\mathbf{x}}_u\) for every \(u\in U\). Against the preceding backdrop, \eqref{eq: ADMM_General_Updates} can be re-written as 
\begin{subequations}\label{eq: ADMM_General_Final}
 \begin{align}
    \bar{\mathbf{x}}^{(n+1)}_u &=  \argmax_{\bar{\mathbf{x}}_u, \ \text{s.t.}~\eqref{eq: PathConstraintsM}-\eqref{eq: Pos_constraintM}~\text{hold}} \quad (w_u + \omega_u) \log (f_u (\bar{\mathbf{x}}_u)) 
- p_{\outsideOpt}x_{u,\outsideOpt} 
-{\mathbf{p}}_u^{(n)}{}^\top\mathbf{x}_u \notag   \\
&\hspace{3cm}
-  \lambda_u^{(n)} (\tilde{\mathbf{a}}_u^\top \mathbf{x}_u + x_{u, \varnothing} - 1) - \frac{\beta}{2} \|\mathbf{x}_u - \mathbf{y}_u^{(n)}\|^2 \label{subeq: ADMM_General_x_updated}
\\
({\mathbf{y}}^{(n+1)},\mathbf{z}^{(n+1)})
&= \argmax_{{\mathbf{y}}\in \mathbb{R}^{U \times |\tilde{\mathcal{E}}|}, \ {\mathbf{z}} \in \mathbb{R}^{|\tilde{\mathcal{E}}|}_+} - \mathbf{p}^{(n)}{}^\top \mathbf{z} - \frac{\beta}{2} \sum_{u \in U} \|\mathbf{x}_u^{(n+1)} - \mathbf{y}_u\|^2 - \frac{\beta}{2} \left\| \sum_{u \in U} \mathbf{y}_u + \mathbf{z} - \ell \right\|^2 \label{subeq: ADMM_General_y_updated}
\\ 
    \lambda^{(n+1)}_u &= \lambda_u + \beta(\tilde{\mathbf{a}}_u^\top\mathbf{x}_u^{(n+1)} + x_{u,\varnothing}^{(n+1)} - 1), \quad \forall \ u\in U \label{subeq: ADMM_General_lam_updated}\\ 
    \mathbf{p}^{(n+1)} &= \mathbf{p}^{(n)}_u + \beta(\sum_{u\in U}\mathbf{y}_u^{(n+1)} + \mathbf{z}^{(n+1)}-\ell) \label{subeq: ADMM_General_p_updated}
    \end{align}
    \end{subequations}

Updates \eqref{eq: ADMM_General_Final} correspond to the inner loop updates in Algorithm \ref{algo:1}, where \eqref{subeq: ADMM_General_x_updated} is implemented locally by different AAM vehicles and \eqref{subeq: ADMM_General_y_updated}-\eqref{subeq: ADMM_General_p_updated} are implemented by service provider.

%% file: Appendix/Boards.tex
In this section, we study a scenario of vertiport reservation for (hypothesized) air taxi services in Northern California. We simulate a scenario where different air taxis request access to air routes to transport people at an urban and regional level. The vertiports in this simulation are located in various cities in Northern California as shown in the map in Fig. \ref{fig:vertiport-map}. For simplicity, we are modeling linear trajectories and assuming a maximum travel range of 100 miles.

\begin{figure}
    \centering
    \includegraphics[scale=0.4]{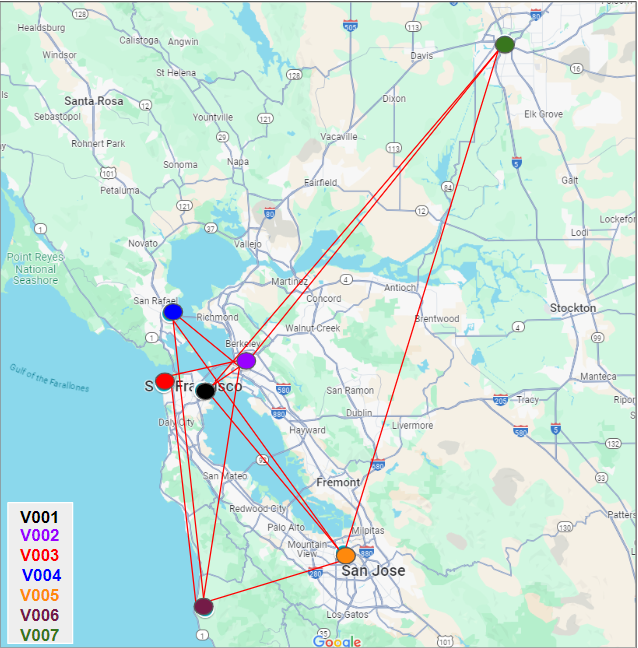}
    \caption{Northern California Vertiport Map. This map, adapted from a Google Maps image, highlights seven distinct vertiports using unique color codes and displays the example routes as red lines.}   
    \label{fig:vertiport-map}
\end{figure}

In this example, 20 air taxis request a departure, air route, and landing clearances among seven vertiport destinations during a 10-minute auction window. The requests by the air taxis, the final allocation of routes, and payments to the SP are presented in Table \ref{tab: ServiceProviderAccess} along with the maximum capacity in every segment of the desired routes, the utility of the air taxis for a given path, and their initial air credits. We set \(\beta= 50, p_\outsideOpt = 10, v_{u,\outsideOpt} = 1, v_{u,\dropOut}= 1, N = 2, \textsf{tol} = 1\times 10^{-4}\). The \textit{Maximum Capacity} column in Table \ref{tab: ServiceProviderAccess} specifies the maximum number of vehicles that can traverse a travel segment at any given time. These values help the reader identify contested travel segments and understand why certain agents must compete for access. In these tables, air taxis sharing the same color represent those that simultaneously requested the same trajectory slot, leading to a constraint violation. As a result, only a subset of these air taxis were granted their preferred route, while the others were denied access. We also present the rank number of these agents representing the order in which each agent computed their integral allocation, as outlined in Algorithm 2.

Below, we highlight the main observations from our numerical study.  
 \begin{itemize}
     \item[(i)] At time step 16, AC003, AC004, and AC015 request departure from V002, which has a departure capacity constraint of one. Consequently, only AC004 is allocated to depart at this time step due to its higher air credits, while AC003 and AC015 are delayed. Naturally, these air taxis would prefer to depart at the next time step; however, they now compete for departure from V002 with AC013 and AC018 at time steps 17 and 18, respectively. Notably, Algorithm \ref{algo:2} prioritizes AC013 and AC018, resulting in further delays for AC003 and AC015. At time step 19, AC003 and AC015 compete again, with AC003 receiving priority due to its higher air credits in Algorithm \ref{algo:2}.

     \item[(ii)]  AC002 and AC011 request landing slots at V004 at the same time, which results in delay for AC011. This is because AC011 has both a lower budget and lower utility in comparison to AC002.
    \item[(iii)] AC009 and AC010 request departure from V001 at the same time, which results in delay for AC009. This is because AC010 has a significantly higher budget than AC009.
    \item[(iv)] Air taxis that are delayed are charged less than those who are allocated their preferred routes. 
    
 \end{itemize}

\begin{table}[h!]
\centering
\footnotesize
\caption{Results of the allocation of air taxis to the desired routes, payments to the SP, utility, initial air credits, and maximum capacity in the en-route travel segment}
\renewcommand{\arraystretch}{1.2}
\setlength{\tabcolsep}{2pt} 
\resizebox{\textwidth}{!}{ 
\begin{tabular}{|c|c|c|c|c|c|c|c|c|c|}
\hline
\textbf{Aircraft} & \thead{\textbf{Req. Route} \\ \textbf{(Orig., Dest.)}}  & \thead{\textbf{Req. Time} \\ \textbf{(Arr, Dep)}}& \thead{\textbf{Max. Capacity} \\ \textbf{(Dep, Route, Arr)}} & \thead{\textbf{Allocated} \\ \textbf{Time} \\ \textbf{(Arr, Dep)}}& \textbf{Status} & \thead{\textbf{Price} \\ \textbf{(\$)}}  & \thead{\textbf{Initial} \\ \textbf{Air} \\ \textbf{Credits}} & \textbf{Utility} & \textbf{Rank} \\ \hline

AC001 & (V007, V002) & (16, 54)  & (2,4,1) & (16, 54) & on-time & 0.0  & 125 & 118 & 6 \\ \hline

{\color{magenta}AC002}& {\color{magenta}(V005, V004)} &{\color{magenta}(19, 47)} & {\color{magenta}(4,5,1)} & {\color{magenta}(19, 47)} & {\color{magenta}on-time} & {\color{magenta}5.73} & {\color{magenta}90} & {\color{magenta}171} & {\color{magenta}7}
\\ \hline

{\color{magenta}\textbf{AC011}} & {\color{magenta}\textbf{(V006, V004)}}  &{\color{magenta}\textbf{(19, 47)}}& \textbf{{\color{magenta}(1,2,1)} }& {\color{magenta}\textbf{(20, 48)}} & {\color{magenta}\textbf{delayed}} & {\color{magenta}\textbf{1.53}} & \textbf{{\color{magenta}78} }& \textbf{{\color{magenta}135}} & \textbf{{\color{magenta}19}}
\\ \hline

\textbf{{\color{blue}AC003}} &\textbf{{\color{blue}(V002, V001)}}& \textbf{{\color{blue}(16, 21)}} & \textbf{{\color{blue}(1,1,2)}} & \textbf{{\color{blue}(19, 24)}} & {\color{blue}\textbf{delayed}} & {\color{blue}\textbf{3.71}} & \textbf{{\color{blue}135}}& \textbf{{\color{blue}172}}  & \textbf{{\color{blue}18}} \\ \hline

{\color{blue}AC004} & {\color{blue}(V002, V001)} & {\color{blue}(16, 21)} &  {\color{blue}(1,1,2)} & {\color{blue}(16, 21)} & {\color{blue}on-time}  & {\color{blue}20.36 } &{\color{blue}154} & {\color{blue}133} & {\color{blue}13}
\\ \hline
\textbf{{\color{blue}AC015}} & \textbf{{\color{blue}(V002, V001)} }& \textbf{{\color{blue}(16, 21)}} & \textbf{{\color{blue}(1,1,2)}} & \textbf{{\color{blue}(20, 25)}} & \textbf{{\color{blue}\textbf{delayed}}}  & {\color{blue}\textbf{0.86}} & \textbf{{\color{blue}65}} & \textbf{{\color{blue}194}} & \textbf{{\color{blue}20}}
\\ \hline

{\color{blue}AC013}  & {\color{blue}(V002, V006)} & {\color{blue}(17, 41)} & {\color{blue}(1,4,3)}  &  {\color{blue}(17, 41)} & {\color{blue}on-time } & {\color{blue}11.11}  & {\color{blue}55} & {\color{blue}147} & {\color{blue}16}
\\ \hline

{\color{blue}AC018}  & {\color{blue}(V002, V007)} & {\color{blue}(18, 56)} & {\color{blue}(1,2,3)} & {\color{blue}(18, 56)} &  {\color{blue}on-time}  & {\color{blue}8.46} & {\color{blue}103} & {\color{blue}165}  & {\color{blue}8}
\\ \hline

AC005  & (V003, V002) & (11, 19) & (1,5,1) & (11, 19) & on-time & 0.0 & 83 & 177  & 4 \\ \hline

AC006  & (V005, V007) & (18, 68) & (4,3,3) & (18, 68) & on-time & 0.0 & 199 & 148  & 15 \\ \hline

AC007  & (V003, V002) & (15, 23) & (1,5,1) & (15, 23) & on-time & 0.0 & 100 & 183 & 5\\ \hline

AC008  & (V007, V001) & (12, 54) & (2,3,2) & (12, 54) &  on-time & 0.0 & 104 & 155 & 10  \\ \hline

{\color{red}\textbf{AC009}}  & {\color{red}\textbf{(V001, V005)}} & {\color{red}\textbf{(13, 34)}} & \textbf{{\color{red}(5,1,2)} } & {\color{red}\textbf{(14, 35)}} & {\color{red}\textbf{delayed}} & {\color{red}\textbf{2.10}} & \textbf{{\color{red}67}} & \textbf{{\color{red}189}} & \textbf{{\color{red}17}} \\ \hline
{\color{red}AC010} & {\color{red}(V001, V005)} & {\color{red}(13, 34)} & {\color{red}(5,1,2)} & {\color{red}(13, 34)} &{\color{red}on-time}& {\color{red}5.75}  & {\color{red}114} & {\color{red}163} & {\color{red}3}  \\ \hline
AC012 & (V005, V001) & (16, 37) & (4,3,2) & (16, 37) & on-time  & 0.0  & 90 & 124  & 12\\ \hline
AC014  & (V001, V002) & (11, 24) & (5,2,1) &  (11, 24) & on-time& 0.0 & 64 & 174   & 9  \\ \hline
AC016 & (V007, V005) & (17, 67) & (2,5,2) & (17, 67) & on-time  & 0.0 & 109 & 189  & 14 \\ \hline
AC017  & (V004, V006) & (16, 44) & (5,3,3) & (16, 44)  & on-time & 0.0 & 155 & 149  & 11 \\ \hline
AC019  & (V004, V002) & (16, 35) & (5,5,2) & (16, 35)  & on-time & 0.0 & 104 & 147& 1\\ \hline
AC020  & (V003, V006) & (16, 38) & (1,2,3) & (16, 38)  & on-time & 0.11 & 96 & 146 & 2\\ \hline
\end{tabular}
}
\label{tab: ServiceProviderAccess}
\end{table}

%% file: Appendix/TableOfNotations.tex
\begin{center}
\begin{tabular}{||c  p{10cm} ||} 
\hline
Notation & Description  \\ [0.5ex] 
 \hline
\(\mathcal{G}\) & Graph representing the airspace \\
$\tilde{\mc{G}}$ & Graph representing the time-extended airspace \\
$\mc{R}$ & Set of regions/sectors in the urban airspace \\
$ \tilde{\mc{E}}$ & Set of edges indicating feasible movement between contiguous regions \\
$C^{\text{arr}}(r,t), C^{\text{dep}}(r,t), C^{\text{stay}}(r,t)$ & Maximum number of vehicles that can arrive, depart, or stay in the region $r \in R$ at time $t$ \\
$W$ & Auction window time interval \\
$U(t)$ & Set of AAM vehicles arriving in the system at time $t$ \\
$M_u$ & Menu of time-trajectories (air corridors) for AAM vehicle \(u\) \\
$R_u$ & Set of routes in the menu of AAM vehicle \(u\) \\
$\nu(r,t)$ & Node in region $r$ and time $t$ \\
$v_{u,s}$ & Valuation (utility) of vehicle $u$ for path $s$ \\
$p_o, p_e$ & Price of the outside edge option, price of an edge $\in  \tilde{\mc{E}}$\\ 
$x_{u,e}$ & Allocation of an AAM vehicle $u$ to edge $e$ \\
$x_{u,o}$ & Allocation of an AAM vehicle $u$ to outside option \\
\(\mathbf{x}_u\) & The vector \((x_{u,e})_{e\in\tilde{\mathcal{E}}}\) \\ \(\bar{\mathbf{x}}_u\) & The vector \([\mathbf{x}_u^\top, x_{u,\outsideOpt}, x_{u,\varnothing}]^\top\)
\\
$w_u$ & Budget of AAM vehicle $u$ (air credits) \\
$e^\ast(s)$ & Departing edge from the origin region on the route \(s\) 
\\
$\ell_e$ & Supply of edge $ e \in  \tilde{\mc{E}}$  \\
$\bar{\mathbf{x}}_{u}$ & Optimal solution for the IOP with optimal price $\mathbf{p^*}$ \\ 
$f_u(\bar{\mathbf{x}}_u)$ & Total utility of allocation $\bar{x}_{u}$ \\
$\tilde{\mc{E}}^{(1)}, \tilde{\mc{E}}^{(2)}, \tilde{\mc{E}}^{(3)}, \tilde{\mc{E}}^{(4)} $ & Arrival (landing), departing (take-off), stay (parking), and transit (air-corridor) edges
\\
\hline
\end{tabular}
\end{center}